\documentclass{article}

\usepackage{microtype}
\usepackage{graphicx}
\usepackage{subfigure}
\usepackage{booktabs} 

\usepackage{amsmath,amssymb}
\usepackage{graphicx}
\usepackage{wrapfig}
\usepackage{verbatim}
\usepackage{bm}
\usepackage{selectp}

\usepackage{hyperref}

\usepackage{amsthm}
\theoremstyle{property}

\theoremstyle{theorem}
\newtheorem{theorem}{Theorem}[section]

\theoremstyle{corollary}
\newtheorem{corollary}{Corollary}[section]

\theoremstyle{lemma}
\newtheorem{lemma}{Lemma}[section]
 
\theoremstyle{definition}
\newtheorem{definition}{Definition}[section]
 
\theoremstyle{remark}
\newtheorem*{remark}{Remark}

\DeclareMathOperator*{\argmax}{arg\,max}

\usepackage{xcolor}

\newcommand{\jp}[1]{}
\newcommand{\kt}[1]{}
\newcommand{\po}[1]{}
\newcommand{\mr}[1]{}
\newcommand{\so}[1]{}
\newcommand{\jbl}[1]{}
\newcommand{\bdv}[1]{}

\usepackage[accepted]{icml2020}

\icmltitlerunning{Finding Equilibrium via Regularization}

\begin{document}

\twocolumn[
\icmltitle{From Poincar\'e Recurrence to Convergence in Imperfect Information Games: Finding Equilibrium via Regularization}





\icmlsetsymbol{equal}{*}

\begin{icmlauthorlist}

\icmlauthor{Julien Perolat}{dm}
\icmlauthor{Remi Munos}{dm}
\icmlauthor{Jean-Baptiste Lespiau}{dm}
\icmlauthor{Shayegan Omidshafiei}{dm}
\icmlauthor{Mark Rowland}{dm}
\icmlauthor{Pedro Ortega}{dm}
\icmlauthor{Neil Burch}{dm}
\icmlauthor{Thomas Anthony}{dm}
\icmlauthor{David Balduzzi}{dm}
\icmlauthor{Bart De Vylder}{dm}
\icmlauthor{Georgios Piliouras}{sutd}
\icmlauthor{Marc Lanctot}{dm}
\icmlauthor{Karl Tuyls}{dm}
\end{icmlauthorlist}

\icmlaffiliation{dm}{DeepMind}
\icmlaffiliation{sutd}{SUTD}

\icmlcorrespondingauthor{Julien Perolat}{perolat@google.com}

\icmlkeywords{Machine Learning, ICML}

\vskip 0.3in
]



\printAffiliations{}  
\begin{abstract}
In this paper we investigate the Follow the Regularized Leader dynamics in sequential imperfect information games (IIG). We generalize existing results of Poincar\'e recurrence from normal-form games to  zero-sum two-player imperfect information games and other sequential game settings. We then investigate how adapting the reward (by adding a regularization term) of the game can give strong convergence guarantees in monotone games. We continue by showing how this reward adaptation technique can be leveraged to build algorithms that converge exactly to the Nash equilibrium. Finally, we show how these insights can be directly used to build state-of-the-art model-free
algorithms for zero-sum two-player Imperfect Information Games (IIG).
\end{abstract}

\section{Introduction}
This paper addresses the problem of learning a Nash equilibrium in several classes of games. Learning Nash equilibria in competitive games is complex as agents no longer share information but behave independently. Various techniques have been proposed to solve these games, with the current state-of-the-art usually guaranteeing average-time convergence of the learned policy to a Nash equilibrium, but not necessarily convergence of the policy itself to Nash. Unfortunately, these convergence guarantees are not conducive to learning in large games, which rely on general function approximation techniques (e.g., deep neural networks) that are inherently  difficult to time-average. Moreover, the real-time behaviors of the policy can be quite distinctive from its time-average counterpart, and can even diverge away from Nash equilibria~\cite{bailey2018multiplicative}.

In some adversarial games, Follow the Regularized Leader~(FoReL) is known to be convergent if the equilibrium is deterministic, and recurrent if the equilibrium is mixed with full support~\cite{mertikopoulos2018cycles}. 
A special case of FoReL dynamics is replicator dynamics~\cite{TaylorJonkerRD}, the main dynamic of evolutionary game theory, whose recurrent behavior in zero-sum games and generalizations is well studied~\cite{Piliouras14,boone2019darwin}. 
More generally, value-based methods have been well-studied in multi-agent reinforcement learning~\cite{Littman94markovgames, Littman01FFQ,Hu03NashQ} but numerous issues of convergence have been noticed. 
But a notable empirical finding shows that regularization of $Q$-learning in matrix games can induce the policy to converge in real-time to a Nash equilibrium~\cite{Tuyls_Q_learning_RD,KaisersT10,KaisersT11} or in the replicator dynamics to Quantal-Response-Equilibrium~\cite{ortega2018modeling, mckelvey1995quantal,Tuyls05Evolutionary}. Other theoretical investigations show that softmax best response can guarantee convergence in $Q$-learning~\cite{leslie2005individual}.

\begin{figure}
  \vspace{-10pt}
  \begin{center}
    \includegraphics[width=0.4\textwidth]{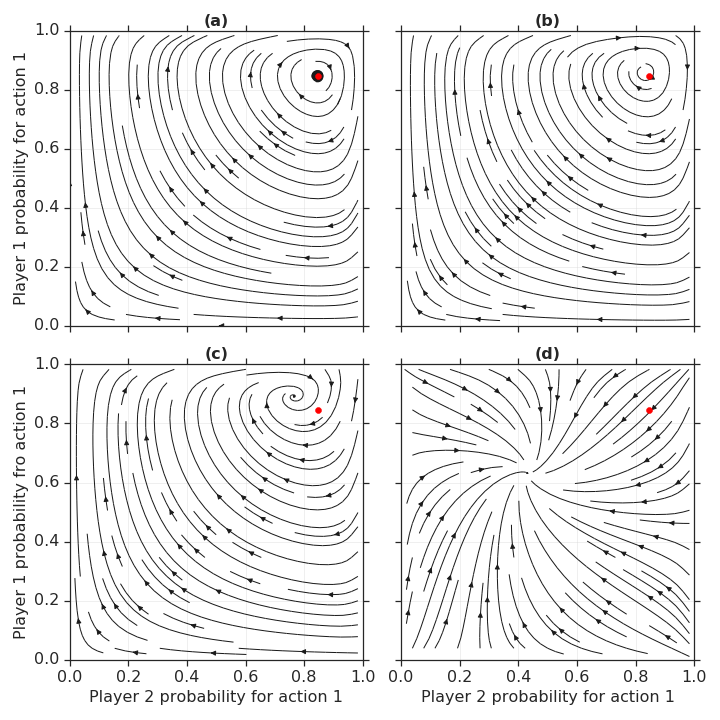}
  \end{center}
  \vspace{-10pt}
  \caption{The trajectory plots for FoReL (plot (a)) and for the version with a reward transform with a parameter $\eta$ multipliers 0.5, 1 and 10 (plots (b), (c), (d) respectively) in a biased matching pennies game (the payoff table for the first player is $[[1, -1], [-1, 10]]$). The red dot is the equilibrium policies of the original game.}
  \label{log-R-RD-fig}
\end{figure}

Motivated by these findings, this paper formally analyzes the impact of regularization on learning dynamics, extending beyond the simple case of matrix games and focusing particularly on the application of FoReL to imperfect information games. 
The contributions of the paper are as follows:
\begin{itemize}
    \item We generalize the Poincar\'e recurrence result~\cite{mertikopoulos2018cycles} to the case of sequential imperfect information games. This proves that strategies can cycle in IIG when using FoReL (e.g., similar to the normal form game case Fig.~\ref{log-R-RD-fig}, (a)).
    \item We prove that changing the reward structure of the game improves convergence guarantees at the cost of slightly modifying the equilibrium of the game (e.g., as in Fig.~\ref{log-R-RD-fig}, (b), (c), (d)).
    \item We show that this reward adaptation method can be used to build a sequence of closer and closer pseudo-solutions converging onto a Nash equilibrium.
    \item We illustrate that by using these theoretical findings, we improve the state-of-the-art of deep reinforcement learning in imperfect information games.
\end{itemize}
\subsection{Related Work}
We discuss related work along three axes: (i) follow the regularized leader and regret minimization, (ii) gradient based methods in differentiable games, and (iii) dynamic programming and reinforcement learning approaches in games.
\paragraph{FoReL and regret minimization in Games.}
There exists a large corpus of literature providing evidence that minimax equilibria (or Nash equilibria) in {\bf normal form zero-sum two-player} games are often unstable rest points of FoReL (or at best neutrally stable). In evolutionary game theory, the replicator dynamics \cite{Zeeman80,Zeeman81,Weibull97,Gintis09} are known to be unstable in the case of an interior equilibrium in zero-sum two-player normal form games~\citep{bloembergen2015evolutionary}.
Many machine learning approaches can be used in self-play to learn an equilibrium: regret minimization methods have been extensively studied in zero-sum games~\cite{PLG, syrgkanis2015fast, Fudenberg98, cfr, Hofbauer09, PLG}, for which the average policy played over time converges to an equilibrium, but the actual policy is known to be recurrent~\cite{Piliouras14, mertikopoulos2018cycles}. The convergence of the actual policy can be obtained when the opponent plays a best response~\cite{waugh2015unified, abernethy2018faster} but not in the self-play setting.
The best response sequences of Fictitious Play ~\cite{Brown51} and smoother variants~\cite{hofbauer2002global} converge in time-average. 
{\bf Polymatrix games } can be solved by linear programming~\cite{cai2016zero} (we will study a generalization of this class). Regret minimization techniques can be used to learn a Nash equilibrium,~\footnote{Since for a coarse correlated equilibrium, the marginals with respect to the players are a Nash equilibrium~\cite{cai2016zero}} but also in this setting, the convergence to a Nash equilibrium requires to compute a time-average policy, and the policy itself is recurrent~\cite{mertikopoulos2018cycles}.
\paragraph{Gradient Based Methods in Differentiable Games.}
Differentiable games (e.g. GANs) trained by gradient descent present many failure modes~\citep{balduzzi2018mechanics}. In~\citep{balduzzi2018mechanics} the authors prove that learning dynamics of gradient descent can cycle in some classes of games. This problem can be resolved by introducing second-order optimization~\citep{balduzzi2018mechanics, Foerster18, mescheder2017numerics, letcher2018stable}, negative momentum~\cite{gidel2018negative} or game theoretic algorithms~\citep{oliehoek2017gangs, grnarova2017online}.

\paragraph{Reinforcement Learning in Games.}
In sequential imperfect information games RL methods have been applied with mild success.
Independent reinforcement learning has many failure modes under these sequential imperfect information settings, as demonstrated in~\citep{Lanctot17PSRO}. In zero-sum sequential imperfect information games, the policy can cycle around the minimax equilibrium without ever converging, even in simple single-state games~\cite{Piliouras14,mertikopoulos2018cycles,Singh00,bloembergen2015evolutionary,bailey2018multiplicative}.
In cooperative settings, players tend to overfit to the opponent while learning, without being able to generalize to other opponents' behaviors~\cite{Matignon2012}. Generally speaking, in the sequential setting, learning in games can be addressed by either approximate dynamic programming in the perfect information case~\citep{lagoudakis2002value, Perolat15, Perolat16, perolat2016softened, perolat2016generalSum, geist2019theory}, regret minimization algorithms~\citep{cfr, lanctot13phdthesis, Lanctot09mccfr} (which suffer from the aforementioned time-averaging problem), best response algorithms~\citep{Heinrich15FSP, Lanctot17PSRO, Heinrich16}, model free reinforcement learning methods~\citep{srinivasan2018actor, Heinrich16} or policy gradient in the worst case~\cite{lockhart2019computing}. 
However, the previous model free RL methods are not flawless: 
Neural Fictitious Self Play (NFSP)~\citep{Heinrich16} maintains two data sets of respectively $600$ and $2000$ times the size of the game, the methods presented in~\citep{srinivasan2018actor} empirically show a convergence in time-average without formal proof, and \citep{lockhart2019computing} require the exact computation of a best response.

\section{Warming up: Normal Form Games}
\label{warm_up}
We first sketch our main results in repeated zero-sum two-player normal form games.
\paragraph{Background.} In a zero-sum two-player normal form game, two players select their actions $a^i \in A$ ($a = (a^1, a^2) = (a^i, a^{-i})$) according to a policy $\pi^i \in \Delta A$ ($\pi = (\pi^1, \pi^2) = (\pi^i, \pi^{-i})$, where $-i$ encodes the opponent of player $i$),  and as a result will receive a reward $r_\pi^i(a^1, a^2)$. The reward is policy-independent ($r^i(a^1, a^2)$) if it is only a function of the actions of the players and not of their policies; policy-independent reward is a standard assumption in the literature. If policy $\pi$ is played we define the $Q$-function to be the expected reward for player $i$ for action $a_i$ (i.e. $Q^i_\pi(a^i) = \mathbb{E}_{a^{-i} \sim \pi^{-i}}[r_\pi^i(a^i, a^{-i})]$) and the value function to be the expected reward (i.e. $V^i_\pi = \mathbb{E}_{a \sim \pi}[r_\pi^i(a)] = \mathbb{E}_{a^i \sim \pi^i}[Q^i_\pi(a^i)]$).

By definition, a policy $\pi^*$ is a {\bf Nash equilibrium} if for all $\pi$ and for all $i$ we have $V^i_{\pi^i, {\pi^*}^{-i}} - V^i_{\pi^*}\leq 0$. In other words, a Nash equilibrium is a joint policy such that no player has an incentive to change its policy if all the other players stick to their policy.

\paragraph{Follow the Regularized Leader (FoReL).} FoReL is an exploration-exploitation algorithm that maximizes the cumulative payoff of the player (exploitation) minus a regularization term (exploration). The continuous time version of this algorithm is defined as follows:
\begin{align*}
    y_t^i(a^i) = \int \limits_{0}^t Q^i_{\pi_s}(a^i) ds\quad \textrm{  and  }\quad \pi^i_t=\argmax_{p \in \Delta A} \Lambda^i(p, y_t^i)
\end{align*}
where $\Lambda^i(p, y) = \langle y, p \rangle - \phi_i(p)$ 
and $\phi_i$ is the regularizer, a function which is assumed to be: (1)~continuous and strictly convex on $\Delta A$ and (2)~smooth on the relative interior of every face of $\Delta A$ (including $\Delta A$ itself). Standard choices of $\phi_i$ include: (1)~entropy $\phi_i(p)=\sum_{a} p(a) \log p(a)$, and (2)~$\ell_2$-norm $\phi_i(p)=\sum_{a} |p(a)|^2$. The choice of the regularizer lead to different dynamics: entropy regularization yields the replicator dynamics and $l_2$-norm regularization yields the projection dynamics~\cite{mertikopoulos2018cycles}.

We will write $\phi_i^*(y) = \max_p \Lambda^i(p, y)$ and we have the property that $\argmax_{p \in \Delta A} \Lambda^i(p, y) = \nabla_y \phi_i^*(y)$ (maximizing argument \citet[p.147]{shalev2012online}).

If $\pi^*$ is a Nash equilibrium, a useful measure of interest that measures the distance to a Nash equilibrium is
{$$J(y) = \sum\limits_{i=1}^2 \big[ \phi_i^*(y_i) - \langle y_i, \pi^*_i \rangle\big].$$}
This quantity (and its generalization introduced in section.~\ref{background}) will be used to construct strong Lyapunov functions in many games of interest. As a warm up, this section will explore these convergence results in the normal form case.

\paragraph{Recurrence.} If the reward is policy-independent and if there exists an interior equilibrium, it is known that the policy under FoReL will be recurrent~\cite{mertikopoulos2018cycles}. We will generalize this result to sequential Imperfect Information Games in section~\ref{recurrence_FoReL}. 
Crucially, this strong negative result indicates that convergence \emph{cannot} be achieved with FoReL in games with a mixed strategy equilibrium, so long as the reward is policy-independent.
Thus, in the rest of the paper, we will explore how to transform the reward by adding a policy-dependent term to guarantee convergence (see section~\ref{reward_transform},\ref{convergence_to_an_equilibrium}).

\paragraph{Reward transformation and convergence in normal-form games.}
Section~\ref{reward_transform} explores the idea of reward transformation by adding a policy dependent term. If the reward is not policy-independent, one can show that:
\begin{align*}
&\frac{d}{dt}J(y) = \sum \limits_{i=1}^2 \underbrace{[V^i_{\pi^i_t, {\pi^*}^{-i}} - V^i_{\pi^*}]}_{\leq 0 \textrm{ because $\pi^*$ is a Nash}}\\
& + \sum \limits_{i=1}^2 \mathbb{E}_{a \sim ({\pi^*}^{i}, \pi_t^{-i}) }[r^i_{{\pi^*}^{i}, \pi_t^{-i}}(a) - r^i_{\pi_t}(a)]
\end{align*}
We later generalize this result in lemma~\ref{lemma_FoReL}.
As an example, consider the following policy dependent reward, which also preserves the zero-sum property for any policy $\mu$ with a full support:
\begin{align*}
    r^i_\pi(a) = r^i(a^i, a^{-i}) - \eta \log \frac{\pi^i(a^i)}{\mu^i(a^i)} + \eta \log \frac{\pi^{-i}(a^{-i})}{\mu^{-i}(a^{-i})}
\end{align*}
Given the above reward, we can show that:
\begin{align*}
&\frac{d}{dt}J(y) = \sum \limits_{i=1}^2 \underbrace{[V^i_{\pi^i_t, {\pi^*}^{-i}} - V^i_{\pi^*}]}_{\leq 0 \textrm{\ because $\pi^*$ is a Nash}} - \eta \sum \limits_{i=1}^2 KL(\pi^{*i}, \pi^{i}_t)
\end{align*}
This inequality ensures that $\pi_t$ will converge to $\pi^*$, the Nash of the game defined by $r^i_\pi(a)$, using Lyapunov arguments. Note that $\pi^*$ will depend on $\mu$ and $\eta$. Transforming the reward improves the convergence property of the game but will shift the equilibrium, a phenomena illustrated in figure~\ref{log-R-RD-fig} where $\phi_i$ is the entropy for all players. Thus, this technique does not directly guarantee convergence to the Nash of the original game. 
We next introduce a technique to adapt the policy-dependent term in the reward, thereby guaranteeing convergence to the actual Nash equilibrium of the game.

\paragraph{Direct Convergence. }
Solving the original game can be achieved by iteratively solving the game with the reward $r^i_{k,\pi}(h, a) = r^i(a^i, a^{-i}) - \eta \log \frac{\pi^i(a^i)}{\pi_{k-1}^i(a^i)} + \eta \log \frac{\pi^{-i}(a^{-i})}{\pi_{k-1}^{-i}(a^{-i})}$ and use the Nash of that game $\pi_k$ to modify the reward of the next game (starting with $\pi_0$ as the uniform policy). The sequence of policies $(\pi_k)_{k\geq 0}$ converges to $\pi^*$, the equilibrium of the policy-independent reward $r^i(a^i, a^{-i})$. Specifically, we can show that:
{\small
\begin{align*}
&\sum \limits_{i=1}^2 \left[KL(\pi^{*i}, \pi^{i}_k) - KL(\pi^{*i}, \pi^{i}_{k-1})\right]\leq - \sum \limits_{i=1}^2 KL(\pi^{i}_{k}, \pi^{i}_{k-1})
\end{align*}
}

which is enough to prove that $(\pi_k)_{k\geq 0}$ converges to $\pi^*$, using Lyapunov-style arguments (this result is proved in section~\ref{convergence_to_an_equilibrium}). This set of results establishes a foundation for convergent learning in the normal-form case.
We next lay out the principles necessary for generalizing to the IIG setting, with our main result detailed in section~\ref{convergence_to_an_equilibrium}.

\section{Background in Sequential Imperfect Information Games}
\label{background}

In a sequential imperfect information game, $N$ players and a chance player (written $c$) interact sequentially starting from a history $h_{\textrm{init}}$. The set of all possible histories is written $H = \cup_{i \in \{1,\dots,N, c\}} H_i$. The sets $H_i$ are the set of histories at player's $i$ turn (all $H_i$ are disjoint). The set of terminal histories $\mathcal{Z}_i$ is a subset of $H_i$ in which the game has ended ($\mathcal{Z} = \cup_{i \in \{1,\dots,N, c\}} \mathcal{Z}_i$). In each history $h\in H$, the current player will observe an information state $x \in \mathcal{X}=\cup_{i \in \{1,\dots,N, c\}} \mathcal{X}_i$. The function $\tau(h) \mapsto i \in \{1, \dots,N, c\}$ provides the player's turn at a given history. We will also write $x(h) \in \mathcal{X}$ for the information state corresponding to an history $h$. We will write $h \in x$ if $x(h) = x$.

At each history $h\in H\backslash\mathcal{Z}$, the current player will play an action $a \in A$. As a result, each player $i \in \{1, \dots, N\}$ will receive a reward $r^i(h, a)$ and the state will transition to $h'=ha$. We will write $h \sqsubset h'$ if there exists a sequence of $k$ actions $(a_i)_{0 \leq i \leq k}$ such that $h a_0 \cdots a_k =h'$. The history $h$ is then said to be a prefix of $h'$.

A policy $\pi(a|x)$ maps an information state $x$ to a distribution over actions $\Delta A$. The restriction of $\pi$ over $\mathcal{X}_i$ is written $\pi^i$ and $\pi^{-i}$ is the restriction of $\pi$ over $\mathcal{X}\backslash \mathcal{X}_i$. We will write $\pi = (\pi^i, \pi^{-i})$.
As in section.~\ref{warm_up}, we consider a {\bf policy dependent reward} (written $r_\pi^i(h,a)$), which can be dependent on the full policy. The rest of this section introduces reinforcement learning tools used to define FoReL in IIG and used in the proofs.

\paragraph{Value function on the histories.}
The value of a policy for player $i$ at history $h$ is defined as follow:
\begin{align*}
    &V^i_{\pi} (h) = \mathbb{E}\big[\sum \limits_{n \geq 0} r_\pi^i(h_n,a_n) | h_0=h, \; h_{n+1}=h_n \; a_n,\\
    &\; a_n \sim \pi(.|x(h_n))\big]= \sum_a \pi(a|x(h))\left[r_\pi^i(h,a) + V^i_{\pi} (ha)\right]
\end{align*}

The value of a policy for player $i$ at history $h$ while taking action $a$ is defined as follow:
\begin{align*}
    &Q^i_{\pi} (h, a) = \mathbb{E}\big[\sum \limits_{n \geq 0} r_\pi^i(h_n,a_n) | h_0=h, \; a_0=a,\\
    & \; h_{n+1}=h_n \; a_n,\; a_n \sim \pi(.|x(h_n))\big]= r_\pi^i(h,a) + V^i_{\pi} (ha)
\end{align*}
\paragraph{Reach probabilities.}
The reach probability of a history $h$ is (note that this product may include the chance player):
\begin{align*}
    &\rho^{\pi} (h) = \prod \limits_{h'a \sqsubset h} \pi(a|x(h'))
\end{align*}
The reach probability of player $i$ of a history $h$ is:
\begin{align*}
    &\rho^{\pi^i} (h) = \prod \limits_{h'a \sqsubset h, \; \tau(h')=i} \pi(a|x(h'))
\end{align*}
The reach probability of player $-i$ of a history $h$ is (this product may include the chance player too):
\begin{align*}
    &\rho^{\pi^{-i}} (h) = \prod \limits_{h'a \sqsubset h, \; \tau(h') \neq i} \pi(a|x(h'))
\end{align*}
In the end, $\forall h \in H$:$\rho^{\pi} (h) = \rho^{\pi^i} (h) \rho^{\pi^{-i}} (h)$

The reach probability of an information state $x \in \mathcal{X}$ is defined as follows:
$$\rho^{\pi} (x) = \sum \limits_{h \in x} \rho^{\pi} (h) \textrm{ and }\rho^{\pi^{-i}} (x) = \sum \limits_{h \in x} \rho^{\pi^{-i}} (h)$$

Under perfect recall \cite{Zinkevich07New}, we can write for any $h\in x$:
$$\rho^{\pi} (x) = \rho^{\pi^i} (h) \rho^{\pi^{-i}} (x)$$

And under perfect recall we will write for all $h\in x,\;\rho^{\pi^i} (x) = \rho^{\pi^i} (h)$
Furthermore,{
$V^i_{\pi} (h_{\textrm{init}}) = \sum \limits_{h \in H} \rho^{\pi} (h) \sum \limits_{a \in A} \pi(a|x(h)) r^i_{\pi}(h, a)$}

\paragraph{Value Function on the information states.} The only information available to a player is the information state. We define the expected value of the game given such an information state $x$ as follows:
\begin{align*}
    V^i_{\pi} (x) = \frac{\sum \limits_{h \in x} \rho^{\pi} (h) V^i_{\pi} (h)}{\sum \limits_{h \in x} \rho^{\pi} (h)} \underbrace{=}_{\textrm{perfect recall}} \frac{\sum \limits_{h \in x} \rho^{\pi^{-i}} (h) V^i_{\pi} (h)}{\sum \limits_{h \in x} \rho^{\pi^{-i}} (h)}
\end{align*}
And the expected $Q$-function given $x$ and $a$ is:
\begin{align*}
    Q^i_{\pi} (x,a) &= \frac{\sum \limits_{h \in x} \rho^{\pi} (h) Q^i_{\pi} (h,a)}{\sum \limits_{h \in x} \rho^{\pi} (h)}= \frac{\sum \limits_{h \in x} \rho^{\pi^{-i}} (h) Q^i_{\pi} (h,a)}{\sum \limits_{h \in x} \rho^{\pi^{-i}} (h)}
\end{align*}

Now we can define a Nash equilibrium in the sequential imperfect information game setting. Formally:
\begin{definition}
A strategy $\pi$ is a Nash equilibrium if for all $i \in \{1, \dots, N\}$ and for all $\pi'^i$:
$V^i_{\pi'^i, \pi^{-i}} (h_{\textrm{init}}) \leq V^i_{\pi^i, \pi^{-i}} (h_{\textrm{init}})$
\end{definition}
\subsection{Monotone Games}
In this paper, we are interested in: (i) zero-sum two-player games, i.e., $V^1_{\pi} = -V^2_{\pi}$ (many games implemented in OpenSpiel~\cite{lanctot2019openspiel} fall in that category); (ii) in zero-sum $N$-player polymatrix games, i.e., when the value can be decomposed in a sum of pairwise interactions $V^i_{\pi} = \sum \limits_{j \neq i} \tilde V^i_{\pi^i, \pi^j}$ with $\tilde V^i_{\pi^i, \pi^j} = -\tilde V^j_{\pi^j, \pi^i}$ generalizing~\citet{cai2016zero,mertikopoulos2018cycles}; and finally (iii) in games where the profit of one player is decoupled from the interaction with the opponents, i.e., when the value can be decomposed in $V^i_{\pi} = \bar V^i_{\pi^i} + \bar V^i_{\pi^{-i}}$. All these settings can be captured by the following monotonicity condition:
\begin{definition}
Let us define {$\Omega^i(\pi, \mu) = V^i_{\pi^i, \pi^{-i}} (h_{\textrm{init}}) - V^i_{\mu^i, \pi^{-i}} (h_{\textrm{init}}) - V^i_{\pi^i, \mu^{-i}} (h_{\textrm{init}}) + V^i_{\mu^i, \mu^{-i}} (h_{\textrm{init}})$}.
A game is monotone if for all policies $\pi$, $\mu$, $\pi \neq \mu$: 
$$\sum \limits_{i\in \{1,\dots,N\}} \Omega^i(\pi, \mu) \leq 0$$
\end{definition}
This condition is slightly difficult to interpret but as mentioned above, it captures a wide class of games (zero-sum two-player, polymatrix zero-sum games etc.). See proof in appendix~\ref{monotone_games}.

\subsection{Follow the Regularized Leader}
\label{sec-FoReL}
Follow the Regularized Leader in imperfect information games defines a sequence of policies $(\pi_s)_{s \geq 0}$ for all $i \in \{1,\dots,N\}$ and $x\in\mathcal{X}_i$ as follow:
\begin{align*}
    &y_t^i(x, a) = \int \limits_{0}^t \rho^{\pi_s^{-i}}(x) Q^i_{\pi_s}(x, a) ds\\
    & \pi^i_t(.|x)=\argmax_{p \in \Delta A} \Lambda^i(p, y_t^i(x,.))
\end{align*}

We define the following quantity for any Nash equilibrium $\pi^*$ of the game:
{
\begin{align*}
    &J(y) = \sum \limits_{i=1}^N \sum \limits_{x \in \mathcal{X}_i} \rho^{{\pi^*}^i}(x) [\phi_i^*(y^i(x, .)) - \langle \pi^*(.|x), y^i(x, .) \rangle]
\end{align*}
}
This quantity will be at the center of our analysis of Follow the Regularized Leader in sections~\ref{recurrence_FoReL} and~\ref{reward_transform}. The following lemma shows how this quantity evolves if both players learn using Follow the Regularized Leader updates. We will use this quantity to create a Lyapunov function for policy-dependent reward and use it to bound the trajectories of FoReL to prove Poincar\'e recurrence; intuitively, that ``most'' trajectories do not converge to equilibria.
\begin{lemma}
\label{lemma_FoReL}
If $y_t$ is defined as the follow the regularized leader dynamics we have:
\begin{align*}
&\frac{d}{dt}J(y) = \sum \limits_{i=1}^N \underbrace{[V^i_{\pi^i_t, {\pi^*}^{-i}} - V^i_{\pi^*}]}_{\leq 0} + \underbrace{\sum \limits_{i=1}^N \Omega^i(\pi, \pi^*)}_{\textrm{$\leq 0$ for a monotone game}}\\
& + \sum \limits_{i=1}^N \sum \limits_{h \in H\backslash\mathcal{Z}} \rho^{\pi_t^{-i}}(h) \rho^{{\pi^*}^{i}}(h) \times\\
&\qquad\qquad \mathbb{E}_{a \sim ({\pi^*}^{i}, \pi_t^{-i}) (..|x(h))}[r^i_{{\pi^*}^{i}, \pi_t^{-i}}(h,a) - r^i_{\pi_t}(h,a)]
\end{align*}

(proof in appendix~\ref{proof_lemma_FoReL})
\end{lemma}

\section{Recurrence of FoReL}
\label{recurrence_FoReL}
This section generalizes the results of~\cite{mertikopoulos2018cycles} to Follow the Regularized Leader in monotone Imperfect Information Games when the reward is policy-independent (all the zero-sum two-player games implemented in OpenSpiel~\cite{lanctot2019openspiel} have this property) and when the equilibrium has a full support. This requires two steps, first we will prove that an equivalent learning dynamic is {\bf Divergence-free} (or preserves volume). Then we will use lemma~\ref{lemma_FoReL} to show that all trajectories of this new dynamical system are {\bf bounded}. This is enough to prove that the trajectories of FoReL are {\bf Poincar\'e recurrent}. Intuitively this means that all trajectories will go back to a neighborhood of their starting point arbitrarily often. The Poincar\'e recurrence theorem~\cite{Piliouras14,
mertikopoulos2018cycles, poincare1890probleme} states:
\begin{theorem}
\label{poincare_recurrent}
If a flow preserves volume (is Divergence-free) and has only bounded orbits then for each open set there exist orbits that intersect the set infinitely often.
\end{theorem}

Instead of studying the original dynamical system, we will fix an action $a_x$ for all $x$ and consider the dynamical system (as this system keeps $w$ bounded):
\begin{align}
    &\dot w_t^i(x, a) = \rho^{\pi_t^{-i}}(x) [Q^i_{\pi_t}(x, a) - Q^i_{\pi_t}(x, a_x)]\label{bounded-FoReL1}\\
    & \pi^i_t(.|x)=\argmax_{p \in \Delta A} \Lambda^i(p, w_t^i(x,.))\label{bounded-FoReL2}
\end{align}

\paragraph{Divergence-free.}
In order to get qualitative results on FoReL, we will prove that the FoReL dynamic is Divergence-free (a generalization of a result from~\citet{mertikopoulos2018cycles}).

\begin{lemma}
\label{FoReL-incompressible}
The system defined above (equation~\eqref{bounded-FoReL1} and~\eqref{bounded-FoReL2}) is autonomous (can be written as $\dot w_t = \xi(w_t)$), Divergence-free, and the dynamic of the policy $\pi_t$ is equivalent to the one defined in section~\ref{sec-FoReL} when the reward is policy-independent.
(Proof in appendix~\ref{appendix-FoReL-incompressible})
\end{lemma}
This property is critical as it implies that the dynamical system has no attractor~\citep[p.252, prop 6.6]{Weibull97}.

\begin{remark}
This does not mean that the policy will not converge. If the $w$ diverges, the policy might converge to a deterministic strategy. However, if the Nash is of full support, it will not be an attractor of the dynamical system.
\end{remark}

We now know that Nash equilibria cannot be attractors of FoReL as the system is Divergence-free. In order to prove the Poincar\'e recurrence, we need to prove an additional property. We need the trajectory $w_t$ to remain bounded if the equilibrium $\pi^*$ is interior.

\begin{lemma}
\label{interior_Nash}
If the equilibrium is interior, then $\sum \limits_{i=1}^N [V^i_{\pi^i_t, {\pi^*}^{-i}} - V^i_{\pi^*}] = 0$.

(Proof in appendix~\ref{proof_interior_Nash})
\end{lemma}

\begin{corollary}
\label{H_bounded}
In a monotone game with a policy-independent reward and an interior equilibrium, if $y_t$ is defined as following the FoReL algorithm we have:
$\frac{d}{dt}J(y) \leq 0$
\end{corollary}
Corolary~\ref{H_bounded} implies that the trajectories of the dynamics equation~\eqref{bounded-FoReL1} and~\eqref{bounded-FoReL2} are bounded. This can be proven by directly using arguments from~\citep[Lemma D.2.]{mertikopoulos2018cycles}.

\paragraph{Poincar\'e recurrence.}
As we have seen in the two previous paragraphs, the flow of FoReL is Divergence-free and all trajectories are bounded in the case of monotone games with an interior Nash equilibrium. Thus all orbits are Poincar\'e recurrent.

\section{Reward Transformation and Convergence in IIG}
\label{reward_transform}
In section~\ref{recurrence_FoReL} we have seen that a policy-independent reward signal can lead to recurrent behavior. The idea we study here is to slightly modify the reward signal such that the Nash equilibrium of this new game is an attractor. We will show two reward transformations that guarantee convergence to a Nash equilibrium (with Lyapunov arguments). The first reward transformation applies generally to monotone games and the second one applies specifically to zero-sum games. But first we briefly recall the Lyapunov method.
\paragraph{Lyapunov method.} The idea of the Lyapunov method to study the ordinary differential equation $\frac{d}{dt}y_t = \xi(y_t)$ is to look at the variations of a quantity $\mathcal{F}(y)\geq 0$ (and $\mathcal{F}(y^*)=0$). The function $\mathcal{F}$ is said to be a strict Lyapunov function if:
$$\forall y \neq y^*, \; \frac{d}{dt}\mathcal{F}(y_t) < 0$$
In that case, the $y_t$ will converge to a minimum of $\mathcal{F}$ if $\xi$ is locally Lipschitz and if $\mathcal{F}$ is a continuously differentiable function. The function $\mathcal{F}$ is said to be a strong Lyapunov function if:
$$\frac{d}{dt}\mathcal{F}(y_t) \leq - \beta \mathcal{F}(y_t), \; \beta >0$$
In this case, the $y_t$ will converge to a minimum of $\mathcal{F}$ at an exponentially fast rate $\mathcal{F}(y_t) \leq \mathcal{F}(y_0) \exp(- \beta t)$.
\paragraph{Monotone games.} In the general case of monotone games, the reward that for any $\mu$ preserves the monotonicity is: (see proof in section~\ref{reward_transform_monotonicity}) $$r^i_\pi(h, a) = r^i(h, a) - \frac{\eta \mathbf{1}_{i=\tau(h)}}{\rho^{\pi^{-i}}(h)} \log \frac{\pi(a|x(h))}{\mu(a|x(h))}$$
An immediate corollary of lemma~\ref{lemma_FoReL} is:
\begin{corollary}
\label{corollary_monotone}
In monotone games, the reward transformation $r^i_\pi(h, a) = r^i(h, a) - \frac{\eta \mathbf{1}_{i=\tau(h)}}{\rho^{\pi^{-i}}(h)} \log \frac{\pi(a|x(h))}{\mu(a|x(h))}$ considered above implies that $H$ will be decreasing:

$\frac{d}{dt}J(y)  \leq - \eta \sum \limits_{i=1}^N \sum \limits_{h \in H_i} \rho^{{\pi^*}^i}(h) KL(\pi^*(.|x(h)), \pi_t(.|x(h)))$
\end{corollary}

Finally, if the regularizer $\phi_i$ is the entropy, we can show that the $\Xi(\pi^*, \pi_t) = \sum \limits_{i=1}^N \sum \limits_{h \in H_i} \rho^{{\pi^*}^i}(h) KL(\pi^*(.|x(h)), \pi_t(.|x(h)))$ is a strong Lyapunov function:
\begin{theorem} If the regularizer $\phi_i$ is the entropy:
$$\frac{d}{dt}\Xi(\pi^*,\pi_t)  \leq - \eta \Xi(\pi^*,\pi_t)$$
it implies:
$\Xi(\pi^*,\pi_t) \leq \Xi(\pi^*,\pi_0) \exp(- \eta t)$
(proof in appendix~\ref{proof_strong_lyapunov_function})
\end{theorem}

This method thus introduces a trade-off between the speed of convergence of the algorithm and the transformation we make to the reward (which has an impact on the equilibrium of the transformed game).
\paragraph{Zero-sum two-player games.} 
Whilst the above approach can be applied to all monotone games, the following reward can be applied specifically to zero-sum games. For any $\mu$, this reward keeps the zero-sum property (see appendix~\ref{reward_transform_monotonicity}) and is more prone to sample based methods as the $\frac{1}{\rho^{\pi^{-i}}(h)}$ is not involved,
\begin{align}
    r^i_\pi(h, a) &= r^i(h, a) - \mathbf{1}_{i=\tau(h)} \eta \log \frac{\pi(a|x(h))}{\mu(a|x(h))}\nonumber\\
    &\qquad  + \mathbf{1}_{i\neq\tau(h)} \eta \log \frac{\pi(a|x(h))}{\mu(a|x(h))}\nonumber
\end{align}
And in that case:

\begin{corollary}
\label{corollary_zero_sum}
We have:
\begin{align*}
    &\frac{d}{dt}J(y) \leq\\
    &- \eta \sum \limits_{i=1}^N \sum \limits_{h \in H_i}\rho^{{\pi^*}^i}(h)\rho^{{\pi_t}^{-i}}(h) KL(\pi^*(.|x(h)), \pi_t(.|x(h)))
\end{align*}
\end{corollary}
And here, if the regularizer $\phi_i$ is the entropy, we can show that the $\Xi(\pi^*, \pi_t) = \sum \limits_{i=1}^N \sum \limits_{h \in H_i} \rho^{{\pi^*}^i}(h) KL(\pi^*(.|x(h)), \pi_t(.|x(h)))$ is a strict Lyapunov function:
\begin{theorem} If the regularizer $\phi_i$ is the entropy:
$$\frac{d}{dt}\Xi(\pi^*, \pi_t) \leq  - \eta \zeta \Xi(\pi^*, \pi_t) $$
with $\zeta = \min \limits_{x \in \mathcal{X}}\; \min \limits_{\pi = \argmax_p\Lambda(p,y) \textrm{ and }J(y)\leq J(y_0)} \;\sum \limits_{h \in x} \rho^{{\pi}^{-i}}(h)$

This imply that: $\Xi(\pi^*, \pi_t) \leq \Xi(\pi^*, \pi_0) \exp(- \zeta  \eta t)$
\end{theorem}
\begin{proof}
The proof follows by combining corollary~\ref{corollary_zero_sum} and the result in appendix~\ref{proof_strong_lyapunov_function}.
\end{proof}

In summary, we saw in this section that exponential convergence rates can be achieved in continuous time in imperfect information games using reward transformation.
\begin{remark}
\label{uniqueness_rq}
Corrolary \ref{corollary_monotone} and \ref{corollary_zero_sum} are valid for all Nash of the transformed game. This means that for all $\eta$, the Nash eq. of the transformed game is unique. This uniqueness property is necessary to define the process of the next section.
\end{remark}

\section{Convergence to an Exact Equilibrium}
\label{convergence_to_an_equilibrium}
The previous section introduced a reward transformation (by adding a policy dependent term $r^i_\pi(h, a) = r^i(h, a) - \frac{\eta \mathbf{1}_{i=\tau(h)}}{\rho^{\pi^{-i}}(h)} \log \frac{\pi(a|x(h))}{\mu(a|x(h))}$) to ensure exponential convergence in games. However this method does not ensure convergence to the equilibrium of the game defined on $r^i(h, a)$. In this section, we study the sequence of policies starting from $\pi_0$, being the uniform policy, and $\pi_k$ the solution of the game with the reward transformation $r^i_\pi(h, a) = r^i(h, a) - \frac{\eta \mathbf{1}_{i=\tau(h)}}{\rho^{\pi^{-i}}(h)} \log \frac{\pi(a|x(h))}{\pi_{k-1}(a|x(h))}$. Intuitively, this approach entails that the policy $\pi_k$ will be searched close to the previous iterate $\pi_{k-1}$ (we write $\pi_k = F(\pi_{k-1})$).
\begin{lemma}
\label{iteration_entropy}

Then for any Nash equilibrium of the game $\pi^*$, we have the following identity for the sequence of policy $\pi_k$:
{\small$$\Xi(\pi^*, \pi_k) - \Xi(\pi^*, \pi_{k-1}) = - \Xi(\pi_k, \pi_{k-1}) + \frac{1}{\eta} \sum \limits_{i=1}^N (m_k^i + \delta_k^i + \kappa_k^i)$$
}
Where:
\begin{align*}
    \Xi(\mu,\pi) = \sum \limits_{i=1}^N \sum \limits_{h \in H_i} \rho^{{\mu}^i}(h) KL(\mu(.|x(h)), \pi(.|x(h)))
\end{align*}
Where:
\begin{align*}
    \kappa^i_k &= \sum \limits_{x \in \mathcal{X}^i}\rho^{\pi^{*i}}(x) \rho^{\pi_k^{-i}}(x) \times\\
    &\qquad \sum \limits_{a \in A}\left[\pi^{*i}(a|x(h)) - \pi_k(a|x(h))\right] \;{}^k Q^i_{\pi_k}(x,a)\leq 0
\end{align*}
Where:
$\delta^i_k = V^i_{\pi^{i}_k, \pi^{*-i}}(h_{\textrm{init}}) - V^i_{\pi^*}(h_{\textrm{init}})\leq 0$

And where:
{
\begin{align*}
    m^i_k &= V^i_{\pi_k}(h_{\textrm{init}}) - V^i_{\pi^{*i}, \pi^{-i}_k}(h_{\textrm{init}}) - V^i_{\pi^{i}_k, \pi^{*-i}}(h_{\textrm{init}})\\
    &\quad + V^i_{\pi^*}(h_{\textrm{init}})
\end{align*}
}
And where $\sum \limits_{i=1}^N m^i_k \leq 0$ if the game is monotone (proof in appendix~\ref{proof_iteration_entropy}).
\end{lemma}

\begin{theorem}
In a monotone game with all Nash equilibrium being interior, the sequence of policy $\{\pi_{k}\}_{k\geq 0}$ (or $\{F^k(\pi_{0})\}_{k\geq 0}$) converges to a Nash equilibrium of the game (proof in appendix~\ref{proof_convergence_to_a_Nash}).
\end{theorem}
\begin{remark}
We were only able to prove this result for interior Nash but we conjecture that it is still true for non interior Nash equilibrium.
\end{remark}

\section{Empirical evaluation}
\begin{figure}
  \vspace{-8pt}
  \begin{center}
    \includegraphics[width=0.48\textwidth]{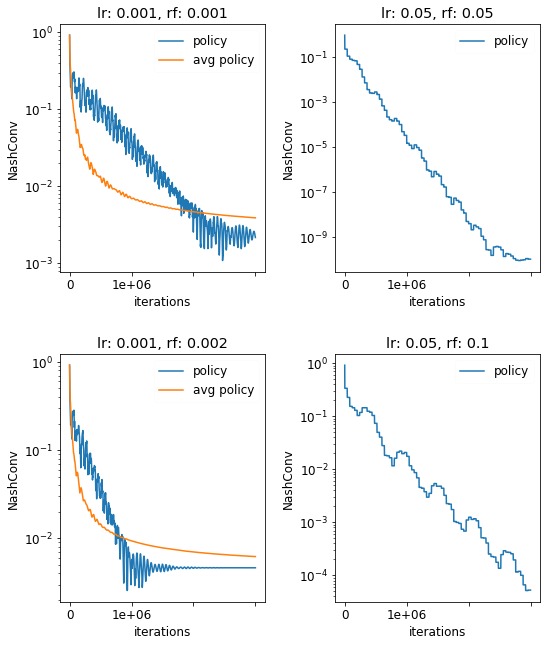}
  \end{center}
  \vspace{-10pt}
  \caption{The left plots illustrate the monotone reward transform whilst the right plot illustrate the direct convergence method shown in section~\ref{convergence_to_an_equilibrium} where we change the reward every 40000 steps (rf is the value of the parameter $\eta$ and lr is the time discretization).}
  \label{Tabular_FoReL}
  \vspace{-1em}
\end{figure}
It has already been empirically noted that regularization helps convergence in games~\cite{omidshafiei2019neural}. Earlier work~\cite{srinivasan2018actor} also provides experiments where the current policy converges in Leduc Poker, whilst the paper only proves convergence analysis of the average policy. Our work sheds a new light on those results as the convergence may have been the result of high regularization (the entropy cost added in~\cite{srinivasan2018actor} appendix G was $0.1$). The experiments will show how reward transform can be used to improve the state of the art of reinforcement Learning in Imperfect Information Games. To keep our implementation as close as possible to FoReL, we use the NeuRD policy update~\cite{omidshafiei2019neural}, a retrace update to estimate the $Q$-function. In order to keep our estimate of the return unbiased, we use that learned $Q$-function as a control variate as in~\cite{schmid2018variance}. The details of the algorithm are in appendix~\ref{empirical_setup}. We present results on four games: Kuhn $\&$ Leduc Poker, Goofspiel and Liars Dice, which have respectively 12, 936, 162 and 24,576 information states. We evaluate all our policies using the NashConv metric~\cite{Lanctot17PSRO} defined as $NashConv(\pi) = \sum_{i=1}^N \max_{\pi^{'i}}V^i_{\pi^{'i},\pi^{-i}}(h_{\textrm{init}})-V^i_{\pi}(h_{\textrm{init}})$.

In this section, we highlight two results with function approximation and illustrate the theory with tabular experiments on Kuhn Poker (figure~\ref{Tabular_FoReL}). A more complete empirical evaluation and the precise description of the setting is available in appendix~\ref{app_experiments}.

\subsection{Experiment with a Decaying Regularization}
We found that decaying the regularization $\eta$ exponentially from an initial value $\eta_{\max} = 1$ to a target value (we looked at values $\{1.0, 0.5, 0.2, 0.05, 0.01, 0.0\}$) is an effective empirical method. In figure~\ref{NeuRD} (top plot), we represent the NashConv as a function of the number of steps. We achieve our best performance for $\eta=0.05$ with a NashConv of $0.10$. This outperforms the results of NFSP~\cite{Heinrich16}, which has a best result of 0.12 in NashConv (0.06 of exploitability reported in the paper) and the state of the art algorithms implemented in Openspiel, which are no better than $0.2$ in NashConv. However, for low choices of $\eta$ the algorithm might diverge.
\subsection{Iteration over the Regularization}
\label{exp_reg_iterative}
As we have seen in section~\ref{convergence_to_an_equilibrium}, the convergence to an exact equilibrium can be achieved by iteratively adapting the the reward. In the experiment  (Fig.~\ref{NeuRD} bottom plot), we change the reward periodically every $N$-steps between steps $[kN, kN + \frac{N}{2}]$ we linearly interpolate between  $r^i_\pi(h, a) = r^i(h, a) - \frac{\eta \mathbf{1}_{i=\tau(h)}}{\rho^{\pi^{-i}}(h)} \log \frac{\pi(a|x(h))}{\pi_{kN}(a|x(h))}$ and $r^i_\pi(h, a) = r^i(h, a) - \frac{\eta \mathbf{1}_{i=\tau(h)}}{\rho^{\pi^{-i}}(h)} \log \frac{\pi(a|x(h))}{\pi_{(k-1)N}(a|x(h))}$ and in interval $[kN + \frac{N}{2}, (k+1)N]$ we use the transformed reward $r^i_\pi(h, a) = r^i(h, a) - \frac{\eta \mathbf{1}_{i=\tau(h)}}{\rho^{\pi^{-i}}(h)} \log \frac{\pi(a|x(h))}{\pi_{kN}(a|x(h))}$. As shown in Fig.~\ref{NeuRD} (bottom plot), this technique allows convergence for very high $\eta$. This is quite an advantage as the method will be more robust to the choice of that hyper-parameter.

\begin{figure}
  \vspace{-8pt}
  \begin{center}
    \includegraphics[width=0.5\textwidth]{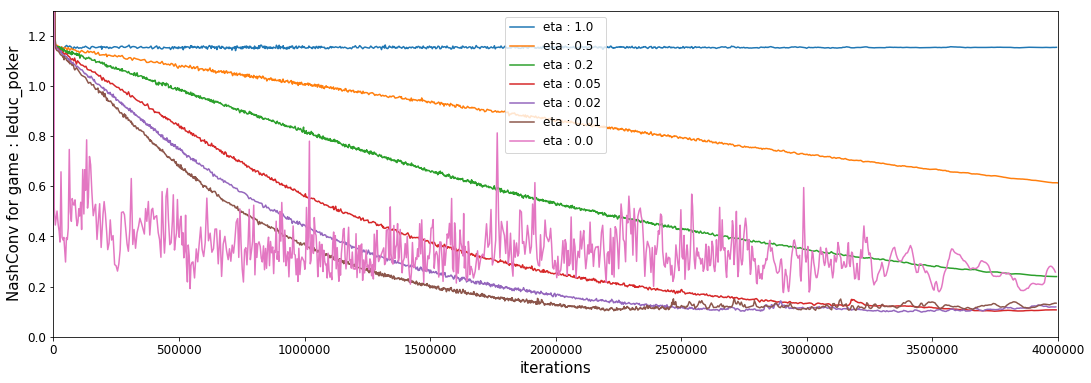}
    \includegraphics[width=0.5\textwidth]{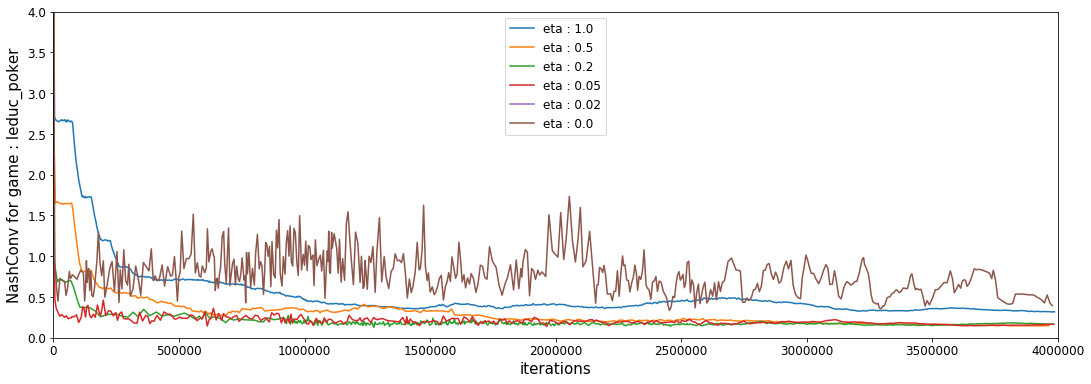}
  \end{center}
  \vspace{-10pt}
  \caption{The precise setup used is described in appendix~\ref{empirical_setup}. The top plot shows results on Leduc poker improving over the NFSP results using a decay of the regularization $\eta$. The bottom plot uses a fixed regularization and adapts the reward every $7.5e4$ steps as described in section~\ref{exp_reg_iterative}.}
  \label{NeuRD}
\end{figure}
\section{Conclusion}
We generalize the Poincar\'e recurrence result for  FoReL from 2-player normal-form zero-sum games to sequential imperfect information games with a monotonicity  condition. Although this is a generalization of a negative convergence result, we show that several reward transformations can guarantee convergence to a slightly modified equilibrium. We also show how to recover the original equilibrium of the game (when it is interior). Finally, based on these techniques we improve the state-of-the-art in model-free deep reinforcement learning in imperfect information games.

Since this work only focuses on FoReL, we aim to analyze the behavior of other dynamics in the sequential case from a dynamical systems perspective in future work. Fictitious play or softmax $Q$-learning have been theoretically considered in normal form games and their analysis with Lyapunov methods remains to be done in the IIG case. Furthermore, the role of regularization for convergence in games needs to be studied more systematically in other settings. Ideas like regularization could also be studied in for example Generative Adversarial Networks.

\bibliography{bib}
\bibliographystyle{icml2020}
\appendix
\newpage
\onecolumn
\section{Proof of Lemma~\ref{lemma_FoReL}}
\label{proof_lemma_FoReL}
\begin{align*}
    \frac{d}{dt}J(y) &= \sum \limits_{i=1}^N \sum \limits_{x \in \mathcal{X}_i} \rho^{{\pi^*}^i}(x) \rho^{\pi_t^{-i}}(x) \langle\pi_t(.|x) - \pi^*(.|x), Q^i_{\pi_t}(x, .) \rangle\\
    &= \sum \limits_{i=1}^N \sum \limits_{x \in \mathcal{X}_i} \rho^{{\pi^*}^i}(x) \sum \limits_{h \in x} \rho^{\pi_t^{-i}}(h) \langle\pi_t(.|x(h)) - \pi^*(.|x(h)), Q^i_{\pi_t}(h, .) \rangle\\
    &= \sum \limits_{i=1}^N \sum \limits_{x \in \mathcal{X}_i} \sum \limits_{h \in x} \rho^{\pi_t^{-i}}(h) \rho^{{\pi^*}^i}(h) \langle\pi_t(.|x(h)) - \pi^*(.|x(h)), Q^i_{\pi_t}(h, .) \rangle\\
    &= \sum \limits_{i=1}^N \sum \limits_{h \in H_i} \rho^{\pi_t^{-i}}(h) \rho^{{\pi^*}^i}(h) \langle\pi_t(.|x(h)) - \pi^*(.|x(h)), Q^i_{\pi_t}(h, .) \rangle
\end{align*}
Let's write ${}^i\bar\pi_t = (\pi^{*i}, \pi^{-i}_t)$ and let's notice that for all $h \in H^{-i}$, we have that ${}^i\bar\pi_t (.|x(h)) = \pi_t (.|x(h))$ and thus for all $h \in H^{-i}$, $\langle\pi_t(.|x(h)) - {}^i\bar\pi_t(.|x(h)), Q^i_{\pi_t}(h, .) \rangle = 0$
\begin{align*}
    \frac{d}{dt}J(y)&=\sum \limits_{i=1}^N \sum \limits_{h \in H} \rho^{\pi_t^{-i}}(h) \rho^{{\pi^*}^i}(h) \langle\pi_t(.|x(h)) - {}^i\bar\pi_t(.|x(h)), Q^i_{\pi_t}(h, .) \rangle\\
    &= \sum \limits_{i=1}^N \Bigg[\sum \limits_{h \in H\backslash\mathcal{Z}} \rho^{{}^i\bar\pi_t}(h) \left[V^i_{\pi_t}(h) - \sum \limits_{a \in A} {}^i\bar\pi_t(a|x(h)) (r^i_{\pi_t}(h,a) + V^i_{\pi_t}(ha))\right] + \sum \limits_{h \in \mathcal{Z}} \rho^{{}^i\bar\pi_t}(h) V^i_{\pi_t}(h) \Bigg]\\
    &=\Bigg[\sum \limits_{i=1}^N \sum \limits_{h \in H} \rho^{{}^i\bar\pi_t}(h) V^i_{\pi_t}(h)\Bigg] - \Bigg[\sum \limits_{i=1}^N \sum \limits_{h \in H\backslash\{h_{\textrm{init}}\}} \rho^{{}^i\bar\pi_t}(h) V^i_{\pi_t}(h)\Bigg] - \Bigg[\sum \limits_{i=1}^N \sum \limits_{h \in H\backslash\mathcal{Z}} \rho^{{}^i\bar\pi_t}(h)\sum \limits_{a \in A} {}^i\bar\pi_t(a|x(h)) r^i_{\pi_t}(h,a)\Bigg]\\
    &=\sum \limits_{i=1}^N V^i_{\pi_t}(h_{\textrm{init}}) - \sum \limits_{i=1}^N \sum \limits_{h \in H\backslash\mathcal{Z}} \rho^{{}^i\bar\pi_t}(h)\sum \limits_{a \in A} {}^i\bar\pi_t(a|x(h)) r^i_{{}^i\bar\pi_t}(h,a)\\
    &\qquad\qquad + \sum \limits_{i=1}^N \sum \limits_{h \in H\backslash\mathcal{Z}} \rho^{{}^i\bar\pi_t}(h)\sum \limits_{a \in A} {}^i\bar\pi_t(a|x(h)) [r^i_{{}^i\bar\pi_t}(h,a)-r^i_{\pi_t}(h,a)]\\
    &= \Bigg[\sum \limits_{i=1}^N V^i_{\pi_t}(h_{\textrm{init}}) - V^i_{(\pi^{*i}, \pi^{-i}_t)}(h_{\textrm{init}})\Bigg] + \sum \limits_{i=1}^N \sum \limits_{h \in H\backslash\mathcal{Z}} \rho^{\pi_t^{-i}}(h) \rho^{{\pi^*}^{i}}(h) \mathbb{E}_{a \sim ({\pi^*}^{i}, \pi_t^{-i}) (..|x(h))}[r^i_{{\pi^*}^{i}, \pi_t^{-i}}(h,a) - r^i_{\pi_t}(h,a)]\\
    &= \Bigg[\sum \limits_{i=1}^N V^i_{(\pi^{i}_t, \pi^{*-i})}(h_{\textrm{init}}) - V^i_{\pi^{*}}(h_{\textrm{init}})\Bigg] + \Bigg[\sum \limits_{i=1}^N V^i_{\pi_t}(h_{\textrm{init}}) - V^i_{(\pi^{*i}, \pi^{-i}_t)}(h_{\textrm{init}}) - V^i_{(\pi^{i}_t, \pi^{*-i})}(h_{\textrm{init}})  + V^i_{\pi^{*}}(h_{\textrm{init}})\Bigg]\\
    &+ \sum \limits_{i=1}^N \sum \limits_{h \in H\backslash\mathcal{Z}} \rho^{\pi_t^{-i}}(h) \rho^{{\pi^*}^{i}}(h) \mathbb{E}_{a \sim ({\pi^*}^{i}, \pi_t^{-i}) (..|x(h))}[r^i_{{\pi^*}^{i}, \pi_t^{-i}}(h,a) - r^i_{\pi_t}(h,a)]
\end{align*}
Which concludes the proof.

\newpage
\section{The system is equivalent to FoReL dynamics and is Divergence-free (lemma \ref{FoReL-incompressible})}
\label{appendix-FoReL-incompressible}
The dynamical system
\begin{align*}
    &\dot w_t^i(x, a) = \rho^{\pi_t^{-i}}(x) [Q^i_{\pi_t}(x, a) - Q^i_{\pi_t}(x, a_x)]\\
    & \pi^i_t(.|x)=\argmax_{p \in \Delta A} \Lambda^i(p, w_t^i(x,.))
\end{align*}
And
\begin{align*}
    &\dot y_t^i(x, a) = \rho^{\tilde \pi_t^{-i}}(x) Q^i_{\tilde \pi_t}(x, a)\\
    & \tilde \pi^i_t(.|x)=\argmax_{p \in \Delta A} \Lambda^i(p, y_t^i(x,.))
\end{align*}
generate the same sequence of policies 

\begin{proof}
For all $x \in \mathcal{X}_i$ the variable:
$$y_t^i(x, a) = \int \limits_{s=0}^t \rho^{\tilde \pi_s^{-i}}(x) Q^i_{\tilde \pi_s}(x, a) ds$$
and we define $\tilde y_t^i(x, a)$ as:
\begin{align*}
    &\tilde y_t^i(x, a)\\
    &= y_t^i(x, a) - y_t^i(x, a_x)\\
    &= \int \limits_{s=0}^t \rho^{\tilde \pi_s^{-i}}(x) Q^i_{\tilde \pi_s}(x, a) ds - \int \limits_{s=0}^t \rho^{\tilde \pi_s^{-i}}(x) Q^i_{\tilde \pi_s}(x, a_x) ds\\
    &= \int \limits_{s=0}^t \rho^{\tilde \pi_s^{-i}}(x) \left[Q^i_{\tilde \pi_s}(x, a) - Q^i_{\tilde \pi_s}(x, a_x) \right] ds
\end{align*}

$$\tilde \pi^i_t(.|x)=\argmax_{p \in \Delta A} \Lambda^i(p, y_t^i(x,.))=\argmax_{p \in \Delta A} = \Lambda^i(p, y_t^i(x,.) - y_t^i(x, a_x)) = \Lambda^i(p, \tilde y_t^i(x, a))$$

Thus $\tilde y_t^i(x, a)$ and $y_t^i(x, a)$ generate the same sequence of policy.

And since $\tilde y_t^i(x, a)$ and $w_t^i(x, a)$ follow the same differential equation and have the same initial conditions, $w_t^i(x, a)$ and $y_t^i(x, a)$ generate the same sequence of policies.
\end{proof}

The dynamical system:
\begin{align*}
    &\dot w_t^i(x, a) = \rho^{\pi_t^{-i}}(x) [Q^i_{\pi_t}(x, a) - Q^i_{\pi_t}(x, a_x)]\\
    & \pi^i_t(.|x)=\argmax_{p \in \Delta A} \Lambda^i(p, w_t^i(x,.))
\end{align*}
is an autonomous dynamical system as $\pi^i_t$ is a function of $w_t^i(x, a)$. 
Let us write it $w_t = \xi (w_t)$ we have $\xi (w_t) (i, x, a) = \rho^{\pi_t^{-i}}(x) [Q^i_{\pi_t}(x, a) - Q^i_{\pi_t}(x, a_x)]$.

Finally, $\forall i \in \{1, \dots, N\}, \forall x \in \mathcal{X}_i, \forall a \in A$, $\xi (w) (i, x, a) = \rho^{\pi^{-i}}(x) [Q^i_{\pi}(x, a) - Q^i_{\pi}(x, a_x)]$ (where $\pi^i(.|x)=\argmax_{p \in \Delta A} \Lambda^i(p, w^i(x,.))$) is independent of $w^i(x, a)$ as $\rho^{\pi^{-i}}(x) Q^i_{\pi}(x, a) = \sum \limits_{h \in x} [r^i(h, a) + V^i_{\pi}(ha)]$ does not depend on $\pi^i(.|x)$.

Thus we have $\frac{\partial \xi (w) (i, x, a)}{\partial w^i(x, a)} = 0$. This proves that the $div_w \xi (w) = \sum \limits_{i=1}^N \sum \limits_{x \in \mathcal{X}_i} \sum \limits_{a \in A} \frac{\partial \xi (w) (i, x, a)}{\partial w^i(x, a)} = 0$ and that the dynamics is incompressible.

\newpage
\section{Proof Strong Lyapunov Function}
\label{proof_strong_lyapunov_function}

\begin{align*}
    &J(y) + \sum \limits_{i=1}^N \sum \limits_{x \in \mathcal{X}_i} \rho^{{\pi^*}^i}(x) \phi_i(\pi^*(.|x)) = \sum \limits_{i=1}^N \sum \limits_{x \in \mathcal{X}_i} \rho^{{\pi^*}^i}(x) [\phi_i^*(y^i(x, .)) - \langle \pi^*(.|x), y^i(x, .) \rangle + \phi_i(\pi^*(.|x))]\\
    &= \sum \limits_{i=1}^N \sum \limits_{x \in \mathcal{X}_i} \rho^{{\pi^*}^i}(x) [\Lambda^i(\pi^i(.|x), y^i(x, .)) - \langle \pi^*(.|x), y^i(x, .) \rangle + \phi_i(\pi^*(.|x))]\\
    &= \sum \limits_{i=1}^N \sum \limits_{x \in \mathcal{X}_i} \rho^{{\pi^*}^i}(x) [\langle \pi^i(.|x), y^i(x, .) \rangle - \phi_i(\pi^i(.|x)) - \langle \pi^*(.|x), y^i(x, .) \rangle + \phi_i(\pi^*(.|x))]\\
    &= \sum \limits_{i=1}^N \sum \limits_{x \in \mathcal{X}_i} \rho^{{\pi^*}^i}(x) [\phi_i(\pi^*(.|x)) - \phi_i(\pi^i(.|x)) + \langle \pi^i(.|x)-\pi^*(.|x), y^i(x, .) \rangle]
\end{align*}
for all $y^i$ in $\{y^i \; | \sum \limits_{a^i \in A^i} y^i(a^i) = 0\}$, the tangent space of $\Delta A^i$, we have that $\nabla h_i(\nabla h_i^*(y^i)) = y^i$ if $\nabla h_i^*(y^i)$ is in the interior of $\Delta A^i$ (see \cite{hofbauer2002global} for the statement of this property). Thus, for all $y^i$ there exists a $\delta$ such that $\nabla h_i(\nabla h_i^*(y^i)) = y^i + \delta \mathbf{1}$

In the end:
\begin{align*}
    &J(y) + \sum \limits_{i=1}^N \sum \limits_{x \in \mathcal{X}_i} \rho^{{\pi^*}^i}(x) \phi_i(\pi^*(.|x)) = \sum \limits_{i=1}^N \sum \limits_{x \in \mathcal{X}_i} \rho^{{\pi^*}^i}(x) [\phi_i(\pi^*(.|x)) - \phi_i(\pi^i(.|x)) + \langle \pi^i(.|x)-\pi^*(.|x), y^i(x, .) \rangle]\\
    &= \sum \limits_{i=1}^N \sum \limits_{x \in \mathcal{X}_i} \rho^{{\pi^*}^i}(x) [\phi_i(\pi^*(.|x)) - \phi_i(\pi^i(.|x)) + \langle \pi^i(.|x)-\pi^*(.|x), \nabla \phi_i (\nabla \phi_i^*(y^i(x, .))) \rangle]\\
    &= \sum \limits_{i=1}^N \sum \limits_{x \in \mathcal{X}_i} \rho^{{\pi^*}^i}(x) [\phi_i(\pi^*(.|x)) - \phi_i(\pi^i(.|x)) + \langle \pi^i(.|x)-\pi^*(.|x), \nabla \phi_i (\pi^i(.|x)) \rangle]\\
    &= \sum \limits_{i=1}^N \sum \limits_{x \in \mathcal{X}_i} \rho^{{\pi^*}^i}(x) D_{\phi_i}(\pi^*(.|x), \pi^i(.|x))
\end{align*}
Where $D_{\phi_i}$ is the Bregman divergence associated with $\phi_i$. If $\phi_i$ is the entropy, we the following equality $J(y) + \sum \limits_{i=1}^N \sum \limits_{x \in \mathcal{X}_i} \rho^{{\pi^*}^i}(x) \phi_i(\pi^*(.|x)) = \sum \limits_{i=1}^N \sum \limits_{x \in \mathcal{X}_i} \rho^{{\pi^*}^i}(x) KL(\pi^*(.|x), \pi^i(.|x)) = \Xi(\pi^*,\pi)$.

That is why: $$\frac{d}{dt}J(y) = \frac{d}{dt}\Xi(\pi^*,\pi)$$

\newpage
\section{Proof of Lemma~\ref{iteration_entropy}}
\label{proof_iteration_entropy}
Let us write ${}^k r^i_\pi(h, a) = r^i(h, a) - \frac{\eta \mathbf{1}_{i=\tau(h)}}{\rho^{\pi^{-i}}(h)} \log \frac{\pi(a|x(h))}{\pi_{k-1}(a|x(h))}$ and ${}^i \bar\pi_k = (\pi^{*i}, \pi_k^{-i})$
\begin{align*}
    \underbrace{\sum \limits_{h \in H} \rho^{\pi_k}(h) \sum \limits_{a \in A} \pi_k(a|x(h)) {}^k r^i_{\pi_k}(h, a)}_{\textrm{(2)}} &= \underbrace{\sum \limits_{h \in H} \rho^{{}^i \bar\pi_k}(h) \sum \limits_{a \in A} {}^i \bar\pi_k(a|x(h)) \; {}^k r^i_{\pi_k}(h, a)}_{\textrm{(3)}}\\
    & \qquad + \underbrace{\sum \limits_{h \in H} \rho^{\pi_k}(h) \sum \limits_{a \in A} \pi_k(a|x(h)) {}^k r^i_{\pi_k}(h, a) - \sum \limits_{h \in H} \rho^{{}^i \bar\pi_k}(h) \sum \limits_{a \in A} {}^i \bar\pi_k(a|x(h)) \; {}^k r^i_{\pi_k}(h, a)}_{\textrm{(1)}}
\end{align*}

Let us write the value function for the reward ${}^k r^i_\pi(h, a)$ and policy $\pi_k$ will be written ${}^k V^i_{\pi_k}(h) = \sum_a \pi_k(a|x(h))\left[{}^k r_{\pi_k}^i(h,a) + {}^k V^i_{\pi_k} (ha)\right]$

\begin{align*}
    \textrm{(1)} &= {}^k V^i_{\pi_k}(h_{\textrm{init}}) - \sum \limits_{h \in H} \rho^{{}^i \bar\pi_k}(h) \sum \limits_{a \in A} {}^i \bar\pi_k(a|x(h)) \; {}^k r^i_{\pi_k}(h, a)\\
    &=\sum \limits_{h \in H} \rho^{{}^i \bar\pi_k}(h) \;{}^k V^i_{\pi_k}(h) - \sum \limits_{h \in H\backslash \{h_{\textrm{init}}\}} \rho^{{}^i \bar\pi_k}(h) \;{}^k V^i_{\pi_k}(h) - \sum \limits_{h \in H} \rho^{{}^i \bar\pi_k}(h) \sum \limits_{a \in A} {}^i \bar\pi_k(a|x(h)) \; {}^k r^i_{\pi_k}(h, a)\\
    &=\sum \limits_{h \in H} \rho^{{}^i \bar\pi_k}(h) \left[\;{}^k V^i_{\pi_k}(h) - \sum \limits_{a \in A} {}^i \bar\pi_k(a|x(h)) \;{}^k V^i_{\pi_k}(ha)\right] - \sum \limits_{h \in H} \rho^{{}^i \bar\pi_k}(h) \sum \limits_{a \in A} {}^i \bar\pi_k(a|x(h)) \; {}^k r^i_{\pi_k}(h, a)\\
    &=\sum \limits_{h \in H} \rho^{{}^i \bar\pi_k}(h) \left[\;{}^k V^i_{\pi_k}(h) - \sum \limits_{a \in A} {}^i \bar\pi_k(a|x(h))\left[{}^k r^i_{\pi_k}(h, a) + \;{}^k V^i_{\pi_k}(ha)\right]\right]\\
    &=\sum \limits_{h \in H} \rho^{{}^i \bar\pi_k}(h) \sum \limits_{a \in A}\left[\pi_k(a|x(h)) - {}^i \bar\pi_k(a|x(h))\right]\;{}^k\\
    &=\sum \limits_{x \in \mathcal{X}^i} \sum \limits_{h \in x} \rho^{\pi^{*i}}(h)\rho^{\pi_k^{-i}}(h) \sum \limits_{a \in A}\left[\pi_k(a|x(h)) - \pi^{*i}(a|x(h))\right]\;{}^k Q^i_{\pi_k}(h,a)\\
    &=\sum \limits_{x \in \mathcal{X}^i}\rho^{\pi^{*i}}(x) \underbrace{\sum \limits_{h \in x} \rho^{\pi_k^{-i}}(h)}_{\rho^{\pi_k^{-i}}(x)} \sum \limits_{a \in A}\left[\pi_k(a|x(h)) - \pi^{*i}(a|x(h))\right]\;{}^k Q^i_{\pi_k}(h,a) \textrm{  from perfect recall}\\
    &=\sum \limits_{x \in \mathcal{X}^i}\rho^{\pi^{*i}}(x) \left[\sum \limits_{h \in x} \rho^{\pi_k^{-i}}(h)\right] \sum \limits_{a \in A}\left[\pi_k(a|x(h)) - \pi^{*i}(a|x(h))\right] \underbrace{\frac{\sum \limits_{h \in x} \rho^{\pi_k^{-i}}(h)\;{}^k Q^i_{\pi_k}(h,a)}{\sum \limits_{h \in x} \rho^{\pi_k^{-i}}(h)}}_{{}^k Q^i_{\pi_k}(x,a)}\\
    &=\sum \limits_{x \in \mathcal{X}^i}\rho^{\pi^{*i}}(x) \rho^{\pi_k^{-i}}(x) \sum \limits_{a \in A}\left[\pi_k(a|x(h)) - \pi^{*i}(a|x(h))\right] \;{}^k Q^i_{\pi_k}(x,a)\\
    &= -\underbrace{\kappa^i_k}_{\leq 0} \geq 0 \textrm{  as $\pi_k$ is a Nash for the game defined on reward ${}^k r^i_\pi(h, a)$}
\end{align*}
Then:
\begin{align*}
    \textrm{(2)} &= \sum \limits_{h \in H} \rho^{\pi_k}(h) \sum \limits_{a \in A} \pi_k(a|x(h)) r^i(h, a) - \eta \sum \limits_{h \in H^i} \rho^{\pi^i_k}(h) \sum \limits_{a \in A} \pi^i_k(a|x(h)) \log \frac{\pi^i_k(a|x(h))}{\pi^i_{k-1}(a|x(h))}\\
    &= V^i_{\pi_k}(h_{\textrm{init}}) - \eta \sum \limits_{h \in H^i} \rho^{\pi^i_k}(h) KL\left(\pi^i_k(.|x(h)), \pi^i_{k-1}(.|x(h))\right)
\end{align*}
And Finally:
\begin{align*}
    \textrm{(3)} &= \sum \limits_{h \in H} \rho^{{}^i \bar\pi_k}(h) \sum \limits_{a \in A} {}^i \bar\pi_k(a|x(h)) r^i(h, a) - \eta \sum \limits_{h \in H^i} \rho^{\pi^{*i}}(h) \sum \limits_{a \in A} \pi^{*i}(a|x(h)) \log \frac{\pi_k(a|x(h))}{\pi_{k-1}(a|x(h))}\\
    &= V^i_{\pi^{*i}, \pi^{-i}_k}(h_{\textrm{init}}) - \eta \sum \limits_{h \in H^i} \rho^{\pi^{*i}}(h) \left[KL\left(\pi^{*i}(.|x(h)), \pi_{k-1}(.|x(h))\right) - KL\left(\pi^{*i}(.|x(h)), \pi_{k}(.|x(h))\right)\right]
\end{align*}
Now combining $\textrm{(2)} = \textrm{(3)} + \textrm{(1)}$ we have:
\begin{align*}
    &\eta \sum \limits_{h \in H^i} \rho^{\pi^{*i}}(h) KL\left(\pi^{*i}(.|x(h)), \pi_{k}(.|x(h))\right) - \eta \sum \limits_{h \in H^i} \rho^{\pi^{*i}}(h) KL\left(\pi^{*i}(.|x(h)), \pi_{k-1}(.|x(h))\right)\\
    &= V^i_{\pi_k}(h_{\textrm{init}}) - V^i_{\pi^{*i}, \pi^{-i}_k}(h_{\textrm{init}})  + \kappa^i_k - \eta \sum \limits_{h \in H^i} \rho^{\pi^i_k}(h) KL\left(\pi^i_k(.|x(h)), \pi^i_{k-1}(.|x(h))\right)\\
    &= \underbrace{V^i_{\pi_k}(h_{\textrm{init}}) - V^i_{\pi^{*i}, \pi^{-i}_k}(h_{\textrm{init}}) - V^i_{\pi^{i}_k, \pi^{*-i}}(h_{\textrm{init}}) + V^i_{\pi^*}(h_{\textrm{init}})}_{=m^i_k} + \underbrace{V^i_{\pi^{i}_k, \pi^{*-i}}(h_{\textrm{init}}) - V^i_{\pi^*}(h_{\textrm{init}})}_{=\delta^i_k} + \kappa^i_k\\
    &\qquad - \eta \sum \limits_{h \in H^i} \rho^{\pi^i_k}(h) KL\left(\pi^i_k(.|x(h)), \pi^i_{k-1}(.|x(h))\right)
\end{align*}
And finally we have the desired property:
\begin{align*}
    &\sum \limits_{h \in H^i} \rho^{\pi^{*i}}(h) KL\left(\pi^{*i}(.|x(h)), \pi_{k}(.|x(h))\right) - \sum \limits_{h \in H^i} \rho^{\pi^{*i}}(h) KL\left(\pi^{*i}(.|x(h)), \pi_{k-1}(.|x(h))\right)\\
    &= \qquad \frac{1}{\eta}m^i_k + \frac{1}{\eta}\delta^i_k + \frac{1}{\eta}\kappa^i_k - \sum \limits_{h \in H^i} \rho^{\pi^i_k}(h) KL\left(\pi^i_k(.|x(h)), \pi^i_{k-1}(.|x(h))\right)
\end{align*}

And we get the result by summing over the players:

\begin{align*}
    &\sum \limits_{i=1}^N \sum \limits_{h \in H^i} \rho^{\pi^{*i}}(h) KL\left(\pi^{*i}(.|x(h)), \pi_{k}(.|x(h))\right) - \sum \limits_{i=1}^N \sum \limits_{h \in H^i} \rho^{\pi^{*i}}(h) KL\left(\pi^{*i}(.|x(h)), \pi_{k-1}(.|x(h))\right)\\
    &= \qquad \frac{1}{\eta}\sum \limits_{i=1}^N m^i_k + \frac{1}{\eta}\sum \limits_{i=1}^N \delta^i_k + \frac{1}{\eta} \sum \limits_{i=1}^N \kappa^i_k - \sum \limits_{i=1}^N \sum \limits_{h \in H^i} \rho^{\pi^i_k}(h) KL\left(\pi^i_k(.|x(h)), \pi^i_{k-1}(.|x(h))\right)
\end{align*}

\newpage
\section{Proof of Lemma~\ref{interior_Nash}}
\label{proof_interior_Nash}

If the equilibrium is interior, then $\sum \limits_{i=1}^N [V^i_{\pi^i_t, {\pi^*}^{-i}} - V^i_{\pi^*}] = 0$.

\begin{proof}
First let us show that $\forall i, \pi^i$:

\begin{align}
&V^i_{\pi^i, {\pi^*}^{-i}} - V^i_{\pi^*}\\
&= \sum \limits_{x \in \mathcal{X}_i} \rho^{\pi^i}(x) \rho^{{\pi^*}^{-i}}(x) \sum \limits_{a \in A} \left({\pi^*}^{i}(a|x)-\pi^i(a|x) \right) Q^i_{\pi^*}(x,a)
\end{align}

Since $\pi^*$ is a Nash equilibrium we always have $V^i_{\pi^i, {\pi^*}^{-i}} - V^i_{\pi^*} \leq 0$. Let us suppose that there exists an information state $x$ such that $Q^i_{\pi^*}(x,a)$ does not have the same values for all actions and that the equilibrium is of full support. Then a greedy policy on that state $x$ (and $\pi^*$ on the other states) should improve the value for player $i$. This would contradict $\pi^*$ being a Nash equilibrium. This proves that all Q-values $Q^i_{\pi^*}(x,a)$ are equals for every states $x$. Then $\sum \limits_{a \in A} \left({\pi^*}^{i}(a|x)-\pi^i(a|x) \right) Q^i_{\pi^*}(x,a) = 0$ for all states.

This concludes the proof that for all $t$, $\sum \limits_{i=1}^N [V^i_{\pi^i_t, {\pi^*}^{-i}} - V^i_{\pi^*}] = 0$.
\end{proof}

\newpage
\section{Reward Transformation in Monotone Games}
\label{reward_transform_monotonicity}

The reward transformation that can be considered are the following:
$$r^i_\pi(h, a) = r^i(h, a) - \mathbf{1}_{i=\tau(h)} \eta \frac{\log \pi(a|x(h))}{\rho^{\pi^{-i}}(h)}$$
or,
$$r^i_\pi(h, a) = r^i(h, a) - \mathbf{1}_{i=\tau(h)} \eta \frac{\log \pi(a|x(h))}{\rho^{\pi^{-i}}(x(h))}$$
or for any $\mu$, $$r^i_\pi(h, a) = r^i(h, a) - \frac{\eta \mathbf{1}_{i=\tau(h)}}{\rho^{\pi^{-i}}(h)} \log \frac{\pi(a|x(h))}{\mu(a|x(h))}$$
or finally for any $\mu$, $$r^i_\pi(h, a) = r^i(h, a) - \frac{\eta \mathbf{1}_{i=\tau(h)}}{\rho^{\pi^{-i}}(x(h))} \log \frac{\pi(a|x(h))}{\mu(a|x(h))}$$

And in that case we have:

$\frac{d}{dt}J(y)  = \sum \limits_{i=1}^N [V^i_{\pi^i_t, {\pi^*}^{-i}} - V^i_{\pi^*}] + \sum \limits_{i=1}^N \Omega^i(\pi, \pi^*) - \eta \sum \limits_{i=1}^N \sum \limits_{h \in H_i} \rho^{{\pi^*}^i}(h) KL(\pi^*(.|x(h)), \pi_t(.|x(h)))$

\begin{proof}
For the game defined on reward $r^i_\pi(h, a) = r^i(h, a) - \frac{\eta \mathbf{1}_{i=\tau(h)}}{\rho^{\pi^{-i}}(h)} \log \frac{\pi(a|x(h))}{\mu(a|x(h))}$
\begin{align}
   V^i_{\pi} (h_{\textrm{init}}) &= \sum \limits_{h \in H} \rho^{\pi} (h) \sum \limits_{a \in A} \pi(a|x(h)) r^i_{\pi}(h, a)\nonumber\\
   &= \sum \limits_{h \in H} \rho^{\pi} (h) \sum \limits_{a \in A} \pi(a|x(h)) \left[ r^i(h, a) - \frac{\eta \mathbf{1}_{i=\tau(h)}}{\rho^{\pi^{-i}}(h)} \log \frac{\pi(a|x(h))}{\mu(a|x(h))} \right]\nonumber\\
   &= \sum \limits_{h \in H} \rho^{\pi} (h) \sum \limits_{a \in A} \pi(a|x(h)) r^i(h, a) - \underbrace{\eta \sum \limits_{h \in H^i} \rho^{\pi^{i}} (h) \sum \limits_{a \in A} \pi^i(a|x(h)) \log \frac{\pi^i(a|x(h))}{\mu^i(a|x(h))}}_{\textrm{Only depends on $\pi^i$.}}\nonumber
\end{align}

Thus if $\sum \limits_{i=1}^N \Omega^i(\pi,\pi^*) = 0$ for the game defined with reward $r^i(h,a)$, then $\sum \limits_{i=1}^N \Omega^i(\pi,\pi^*) = 0$ for the game defined on reward $r^i_\pi(h, a) = r^i(h, a) - \frac{\eta \mathbf{1}_{i=\tau(h)}}{\rho^{\pi^{-i}}(h)} \log \frac{\pi(a|x(h))}{\mu(a|x(h))}$ the monotonicity is also $\sum \limits_{i=1}^N \Omega^i(\pi,\pi^*) = 0$.

\begin{align*}
   &\sum \limits_{i=1}^N \sum \limits_{h \in H\backslash\mathcal{Z}} \rho^{\pi_t^{-i}}(h) \rho^{{\pi^*}^{i}}(h) \mathbb{E}_{a \sim ({\pi^*}^{i}, \pi_t^{-i}) (..|x(h))}[r^i_{{\pi^*}^{i}, \pi_t^{-i}}(h,a) - r^i_{\pi_t}(h,a)]\\
    &=- \eta \sum \limits_{i=1}^N \sum \limits_{h \in H^i\backslash\mathcal{Z}} \rho^{{\pi^*}^{i}}(h) \sum \limits_{a \in A} \pi^{*i}(a|x(h)) \log \frac{\pi^{*i}(a|x(h))}{\pi^i(a|x(h))}
\end{align*}
The other cases are left in appendix.
\end{proof}

\newpage
\section{Reward Transformation in Zero-Sum Games}
\begin{align}
    r^i_\pi(h, a) &= r^i(h, a) - \mathbf{1}_{i=\tau(h)} \eta \log \pi(a|x(h)) + \mathbf{1}_{i\neq\tau(h)} \eta \log \pi(a|x(h))
\end{align}

or for any $\mu$,
\begin{align}
    r^i_\pi(h, a) &= r^i(h, a) - \mathbf{1}_{i=\tau(h)} \eta \log \frac{\pi(a|x(h))}{\mu(a|x(h))} + \mathbf{1}_{i\neq\tau(h)} \eta \log \frac{\pi(a|x(h))}{\mu(a|x(h))}
\end{align}

And in that case:

$$\frac{d}{dt}J(y) = \sum \limits_{i=1}^N [V^i_{\pi^i_t, {\pi^*}^{-i}} - V^i_{\pi^*}] - \eta \sum \limits_{i=1}^N \sum \limits_{h \in H_i} \rho^{{\pi^*}^i}(h)\rho^{{\pi_t}^{-i}}(h) KL(\pi^*(.|x(h)), \pi_t(.|x(h)))$$

\begin{proof}
The game is still zero-sum so the monotonicity is still zero.

\begin{align*}
   &\sum \limits_{i=1}^N \sum \limits_{h \in H\backslash\mathcal{Z}} \rho^{\pi_t^{-i}}(h) \rho^{{\pi^*}^{i}}(h) \mathbb{E}_{a \sim ({\pi^*}^{i}, \pi_t^{-i}) (..|x(h))}[r^i_{{\pi^*}^{i}, \pi_t^{-i}}(h,a) - r^i_{\pi_t}(h,a)]\\
    &=- \eta \sum \limits_{i=1}^N \sum \limits_{h \in H^i\backslash\mathcal{Z}} \rho^{{\pi^*}^{i}}(h) \rho^{\pi_t^{-i}}(h) \sum \limits_{a \in A} \pi^{*i}(a|x(h)) \log \frac{\pi^{*i}(a|x(h))}{\pi^i(a|x(h))}
\end{align*}
As in that case $r^i_{{\pi^*}^{i}, \pi_t^{-i}}(h,a) - r^i_{\pi_t}(h,a) = 0$ for all $h \in H^{-i}$.

\end{proof}

Discuss biais convergence trade-off.

\newpage
\section{Convergence to a Nash}
\label{proof_convergence_to_a_Nash}
The proof of the convergence to an exact Nash uses similar arguments used to prove convergence for strict Lyapunov functions in the discrete vase. From lemma~\ref{iteration_entropy} we know that for a policy sequence starting from $\pi_0$ being the uniform policy and $\pi_k$ is the solution of the game with the reward transformation $r^i_\pi(h, a) = r^i(h, a) - \frac{\eta \mathbf{1}_{i=\tau(h)}}{\rho^{\pi^{-i}}(h)} \log \frac{\pi(a|x(h))}{\pi_{k-1}(a|x(h))}$. In this section, we will call this map $F$ (and $F(\mu) = \pi^*_\mu$ is the equilibrium of the game defined on $r^i_\pi(h, a) = r^i(h, a) - \frac{\eta \mathbf{1}_{i=\tau(h)}}{\rho^{\pi^{-i}}(h)} \log \frac{\pi(a|x(h))}{\mu(a|x(h))}$). We will show that $\pi_k = F^k(\pi_0)$ will converge to a Nash equilibrium of the game $\pi^*$.

The proof proceeds in 3 steps:
\begin{itemize}
    \item First we prove that $F$ is continuous,
    \item Second we prove that $\min_{\pi^*\in \Pi^*} \Xi(\pi^*, F(\mu)) - \min_{\pi^*\in \Pi^*} \Xi(\pi^*, \mu) < 0$,
    \item The second step is enough to prove that $\min_{\pi^*\in \Pi^*} \Xi(\pi^*, \pi_k)$ converges to a value $c$. The last step proves by contradiction that $c$ can't be anything but $0$.
\end{itemize}

\subsection{Continuity}
\label{continuity_of_F}
The first step is to show that the map $F(.)$ which associate $\mu$ to the Nash equilibrium over the game defined over $r^i_{\mu,\pi}(h, a) = r^i(h, a) - \frac{\eta \mathbf{1}_{i=\tau(h)}}{\rho^{\pi^{-i}}(h)} \log \frac{\pi(a|x(h))}{\mu(a|x(h))}$ is continuous.

Then for all $\mu, \mu'$, we have $r^i_{\mu,\pi}(h, a)-r^i_{\mu',\pi}(h, a) =  - \frac{\eta \mathbf{1}_{i=\tau(h)}}{\rho^{\pi^{-i}}(h)} \log \frac{\mu'(a|x(h))}{\mu(a|x(h))}$

Let us write now $w^*_\mu$ and $w^*_{\mu'}$ the fixed point of the dynamic defined in lemma~\ref{FoReL-incompressible} and $\Xi_\mu$ and $\Xi_{\mu'}$ their corresponding Lyapunov function and $\pi^*_{\mu}$ and $\pi^*_{\mu'}$.

Let us consider that $\tilde w$ follow the following ODE (where $Q^i_{\mu,\pi_t}(x, a)$ is the $Q$-function for reward $r^i_{\mu,\pi}(h, a) = r^i(h, a) - \frac{\eta \mathbf{1}_{i=\tau(h)}}{\rho^{\pi^{-i}}(h)} \log \frac{\pi(a|x(h))}{\mu(a|x(h))}$)
\begin{align*}
    &\dot {\tilde w}_t^i(x, a) = \rho^{\pi_t^{-i}}(x) [Q^i_{\mu,\pi_t}(x, a) - Q^i_{\mu,\pi_t}(x, a_x)]\\
    & \pi^i_t(.|x)=\argmax_{p \in \Delta A} \Lambda^i(p, \tilde w_t^i(x,.))
\end{align*}
Let us suppose furthermore that we start from the equilibrium $\tilde w(0) = w^*_\mu$

Let us examine the variation of $\Xi_{\mu'}(\tilde w(t))$ and write $\pi_t(.|x) = \Gamma(\tilde w(t)(.,x))$:
\begin{align*}
    \frac{d}{dt}\Xi(\pi_{\mu'}^*,\pi_t) &= \sum \limits_{i=1}^N \sum \limits_{x \in \mathcal{X}_i} \rho^{{\pi_{\mu'}^*}^i}(x) \rho^{\pi_t^{-i}}(x) \langle\pi_t(.|x) - \pi_{\mu'}^*(.|x), Q^i_{\mu,\pi_t}(x, .) \rangle\\
    &= \sum \limits_{i=1}^N \sum \limits_{x \in \mathcal{X}_i} \rho^{{\pi_{\mu'}^*}^i}(x) \rho^{\pi_t^{-i}}(x) \langle\pi_t(.|x) - \pi_{\mu'}^*(.|x), Q^i_{\mu',\pi_t}(x, .) \rangle\\
    & \qquad + \sum \limits_{i=1}^N \sum \limits_{x \in \mathcal{X}_i} \rho^{{\pi_{\mu'}^*}^i}(x) \rho^{\pi_t^{-i}}(x) \langle\pi_t(.|x) - \pi_{\mu'}^*(.|x), Q^i_{\mu,\pi_t}(x, .) - Q^i_{\mu',\pi_t}(x, .) \rangle
\end{align*}

Let have $\mu, \mu' \in D_0$ (where $D_0$ is an open set). Furthermore let us suppose that for all $\mu \in D_0$ there exists $\epsilon>0$ such that for all $x \in \mathcal{X}$ and $a \in A$ $\mu(a|x)>\epsilon$.

The function $\log$ is locally Lipschitz of constant $K$ in $D_0$.

As $\pi_t = \pi^*_{\mu}$ we can bound $Q^i_{\mu,\pi_t}(x, .) - Q^i_{\mu',\pi_t}(x, .) \leq \eta T_{\max}\left[\sup_{\mu'' \in D_0}\max_{h \in H_i}\frac{1 }{\rho^{\pi_{\mu''}^{*-i}}(h)}\right]K \|\mu-\mu'\|$

Finally We can have that:
\begin{align*}
    \frac{d}{dt}\Xi(\pi_{\mu'}^*,\pi_t) \leq - \eta \Xi(\pi_{\mu'}^*,\pi_t) + \eta T_{\max}\left[\sup_{\mu'' \in D_0}\max_{h \in H_i}\frac{1 }{\rho^{\pi_{\mu''}^{*-i}}(h)}\right]K \|\mu-\mu'\|
\end{align*}

This imply that $\Xi(\pi_{\mu'}^*,\pi_{\mu}^*)\leq T_{\max}\left[\sup_{\mu'' \in D_0}\max_{h \in H_i}\frac{1 }{\rho^{\pi_{\mu''}^{*-i}}(h)}\right]K \|\mu-\mu'\|$.

This finally imply that the map $\mu \rightarrow \pi^*_{\mu}$ is continuous.

\subsection{The function $\min_{\pi* \in \Pi^*}\Xi(\pi^*, \mu)$ is a strict Lyapunov function}
We have seen that the following equality holds (in lemma~\ref{iteration_entropy}):

$\Xi(\pi^*, \pi_k) - \Xi(\pi^*, \pi_{k-1}) = - \Xi(\pi_k, \pi_{k-1}) + \frac{1}{\eta} \sum \limits_{i=1}^N m_k^i + \frac{1}{\eta} \sum \limits_{i=1}^N \delta_k^i + \frac{1}{\eta} \sum \limits_{i=1}^N \kappa_k^i$

Where:
\begin{align*}
    \Xi(\mu,\pi) = \sum \limits_{i=1}^N \sum \limits_{h \in H_i} \rho^{{\mu}^i}(h) KL(\mu(.|x(h)), \pi(.|x(h)))
\end{align*}
Where:
\begin{align*}
    \kappa^i_k = \sum \limits_{x \in \mathcal{X}^i}\rho^{\pi^{*i}}(x) \rho^{\pi_k^{-i}}(x) \sum \limits_{a \in A}\left[\pi^{*i}(a|x(h)) - \pi_k(a|x(h))\right] \;{}^k Q^i_{\pi_k}(x,a)\leq 0
\end{align*}
Where:
\begin{align*}
    \delta^i_k = V^i_{\pi^{i}_k, \pi^{*-i}}(h_{\textrm{init}}) - V^i_{\pi^*}(h_{\textrm{init}})\leq 0
\end{align*}
And where:
\begin{align*}
    m^i_k = V^i_{\pi_k}(h_{\textrm{init}}) - V^i_{\pi^{*i}, \pi^{-i}_k}(h_{\textrm{init}}) - V^i_{\pi^{i}_k, \pi^{*-i}}(h_{\textrm{init}}) + V^i_{\pi^*}(h_{\textrm{init}})
\end{align*}
And where $\sum \limits_{i=1}^N m^i_k \leq 0$ if the game is monotone.

First let us write $\Pi^*$ the set of Nash equilibrium if the game defined on reward $r^i(h, a)$.

{\bf The Goal of this section is to prove that that $\min_{\pi^*\in \Pi^*} \Xi(\pi^*, F(\mu)) - \min_{\pi^*\in \Pi^*} \Xi(\pi^*, \mu) < 0$ for all $\mu \not\in \Pi^*$}

The first step of our proof is to show if there exists a $k$ such that $\Xi(F(\mu), \mu)=0$, then $F(\mu), \mu \in \Pi^*$.

To do so, we first need to prove a serie of technical lemma.

\begin{lemma}
\label{technical_lemma}
For all $\pi, \pi^*$ and for all $i\in \{1,\dots,N\}$ we have: $$V^i_{\pi^*}(h_{\textrm{init}})-V^i_{\pi^i, \pi^{*-i}}(h_{\textrm{init}}) = \sum \limits_{x \in \mathcal{X}_i} \rho^{\pi^i}(x)\rho^{\pi^{*-i}}(x)\sum_a (\pi^i(a|x)-\pi^{*i}(a|x))Q^i_{\pi^*}(x,a)$$
\end{lemma}
\begin{proof}
Let's write $\bar\pi = (\pi^i, \pi^{*-i})$
\begin{align*}
    &V^i_{\pi^*}(h_{\textrm{init}})-V^i_{\pi^i, \pi^{*-i}}(h_{\textrm{init}})\\
    &= \underbrace{\sum \limits_{h \in H} \rho^{\pi^i}(h)\rho^{\pi^{*-i}}(h)[V^i_{\pi^*}(h) - \sum_{a\in A}\bar \pi(a|x(h))V^i_{\pi^*}(ha)]}_{=V^i_{\pi^*}(h_{\textrm{init}})}- \underbrace{\sum \limits_{h \in H} \rho^{\pi^i}(h)\rho^{\pi^{*-i}}(h)\sum_{a\in A}\bar \pi(a|x(h))r^i(h,a)]}_{=V^i_{\pi^i, \pi^{*-i}}(h_{\textrm{init}})}\\
    &= \sum \limits_{h \in H} \rho^{\pi^i}(h)\rho^{\pi^{*-i}}(h)[V^i_{\pi^*}(h) - \sum_{a\in A}\bar \pi(a|x(h))[r^i(h,a) + V^i_{\pi^*}(ha)]]\\
    &= \sum \limits_{h \in H} \rho^{\pi^i}(h)\rho^{\pi^{*-i}}(h)\sum_{a\in A}(\pi^*(a|x(h))-\bar \pi(a|x(h)))[Q^i_{\pi^*}(h,a)]\\
    &= \sum \limits_{h \in H_i} \rho^{\pi^i}(h)\rho^{\pi^{*-i}}(h)\sum_{a\in A}(\pi^{*i}(a|x(h))-\pi^i(a|x(h)))Q^i_{\pi^*}(h,a) \textrm{ as $\bar\pi = \pi*$ on all the opponent nodes.}\\
    &= \sum \limits_{x \in \mathcal{X}_i} \rho^{\pi^i}(x)\rho^{\pi^{*-i}}(x)\sum_{a\in A}(\pi^{*i}(a|x)-\pi^i(a|x))[Q^i_{\pi^*}(x,a)] 
\end{align*}
\end{proof}

\begin{lemma}
\label{technical_lemma_Nash}
Let $\pi^*$ be a policy. If for all $i \in \{1,\dots,N\}, x \in \mathcal{X}_i$ and $\hat \pi$ such that $\sum_a \left(\pi^*(a|x)-\hat\pi(a|x)\right)Q^i_{\pi^*}(x,a)\geq0$ then $\pi^*$ is a Nash equilibrium.
\end{lemma}
\begin{proof}
Let suppose that  for all $i \in \{1,\dots,N\}, x \in \mathcal{X}_i$ and $\hat \pi$ we have $\sum_a \left(\pi^*(a|x)-\hat\pi(a|x)\right)Q^i_{\pi^*}(x,a)\geq0$

Then by lemma~\ref{technical_lemma} we have: $$\forall i, \; V^i_{\pi^*}(h_{\textrm{init}})-V^i_{\pi^i, \pi^{*-i}}(h_{\textrm{init}}) = \sum \limits_{x \in \mathcal{X}_i} \rho^{\pi^i}(x)\rho^{\pi^{*-i}}(x)\sum_a (\pi^i(a|x)-\pi^{*i}(a|x))Q^i_{\pi^*}(x,a)\geq0$$

Thus $\pi^*$ is a Nash equilibrium.
\end{proof}

\begin{corollary}
If $\pi^*$ is not a Nash equilibrium, then there exists $i \in \{1,\dots,N\}, x \in \mathcal{X}_i$ and $\hat \pi$ such that $\sum_a \left(\pi^*(a|x)-\hat\pi(a|x)\right)Q^i_{\pi^*}(x,a)<0$
\end{corollary}
\begin{proof}
This is a direct consequence of lemma~\ref{technical_lemma_Nash}
\end{proof}

\begin{theorem}
\label{theorem_nash_div}
If $\pi^*_\mu = F(\mu) = \mu$, then $\mu$ is a Nash equilibrium of the game defined on $r^i(h,a)$.
\end{theorem}
\begin{proof}
First we will write $V^i_{\mu,\pi}(h)$ ($Q^i_{\mu,\pi}(h,a)$) to be the value function ($Q$-function) with respect to the reward $r^i_{\mu,\pi}(h, a) = r^i(h, a) - \frac{\eta \mathbf{1}_{i=\tau(h)}}{\rho^{\pi^{-i}}(h)} \log \frac{\pi(a|x(h))}{\mu(a|x(h))}$.

The reader will notice that since $\pi^*_\mu = \mu$ then $Q^i_{\mu,\pi^*_\mu}(h,a) = Q^i_{\pi^*_\mu}(h,a)$.

We will prove the result by contradiction. Let suppose that $\pi^*_\mu$ is not a Nash equilibrium for the game with reward $r^i(h,a)$. Then there exists $i$, $\hat \pi$ and $\tilde x\in \mathcal{X}_i$ such that $\sum_a \left(\pi^*_\mu(a|\tilde x)-\hat\pi(a|\tilde x)\right)Q^i_{\pi^*_\mu}(\tilde x,a)<0$

For the rest of this proof, we will write $\hat \pi_\alpha$ the policy defined as $\pi^*_\mu$ on all $x \in \mathcal{X}\backslash\{\tilde x\}$ and $(1-\alpha)\pi^*_\mu + \alpha \hat\pi$ on state $\tilde x$.

As $\pi^*_\mu$ is a Nash equilibrium for the reward $r^i_{\mu,\pi}(h, a)$, then $V^i_{\mu,\pi^*_\mu}(h_{\textrm{init}})-V^i_{\mu,\hat \pi_\alpha}(h_{\textrm{init}}) \geq 0$
\begin{align*}
    &V^i_{\mu,\pi^*_\mu}(h_{\textrm{init}})-V^i_{\mu,\hat \pi_\alpha}(h_{\textrm{init}})\\
    &=\sum \limits_{h \in H} \rho^{\hat \pi_\alpha}(h)[V^i_{\mu,\pi^*_\mu}(h) - \sum_{a\in A}\hat \pi_\alpha(a|x(h))V^i_{\mu,\pi^*_\mu}(ha)] - \sum \limits_{h \in H} \rho^{\hat \pi_\alpha}(h)\sum_{a\in A}\hat \pi_\alpha(a|x(h))r^i_{\mu,\hat \pi_\alpha}(h, a)]\\
    &=\sum \limits_{h \in \tilde x} \rho^{\hat \pi_\alpha}(h)[V^i_{\mu,\pi^*_\mu}(h) - \sum_{a\in A}\hat \pi_\alpha(a|x(h))[r^i_{\mu,\hat \pi_\alpha}(h, a) + V^i_{\mu,\pi^*_\mu}(ha)]] \textrm{ as $\hat \pi_\alpha$ and $\pi^*_\mu$ are only different on $\tilde x$}\\
    &=\sum \limits_{h \in \tilde x} \rho^{\pi^*_\mu}(h)[V^i_{\mu,\pi^*_\mu}(h) - \sum_{a\in A}\hat \pi_\alpha(a|x(h))[Q^i_{\mu,\pi^*_\mu}(h,a) - \frac{\eta \mathbf{1}_{i=\tau(h)}}{\rho^{\pi_\mu^{*-i}}(h)} \log \frac{\hat \pi_\alpha(a|x(h))}{\mu(a|x(h))}]]\\
    &=\sum \limits_{h \in \tilde x} \rho^{\pi^*_\mu}(h) [\frac{\eta \mathbf{1}_{i=\tau(h)}}{\rho^{\pi_\mu^{*-i}}(h)}KL(\hat \pi_\alpha(.|x(h)), \mu(.|x(h))) + \sum_{a\in A}[\pi^*_\mu(a|x(h))-\hat \pi_\alpha(a|x(h))]Q^i_{\mu,\pi^*_\mu}(h,a)]\\
    &=\eta \left(\sum \limits_{h \in \tilde x} \rho^{\pi^{*i}_\mu}(h)\right) KL(\hat \pi_\alpha(.|\tilde x), \mu(.|\tilde x)) + \rho^{\pi^*_\mu}(\tilde x)\sum_{a\in A}[\pi^*_\mu(a|\tilde x)-\hat \pi_\alpha(a|\tilde x)]Q^i_{\mu,\pi^*_\mu}(\tilde x,a)]\\
    &\leq \eta \left(\sum \limits_{h \in \tilde x} \rho^{\pi^{*i}_\mu}(h)\right) \frac{1}{2}\|\hat \pi_\alpha(.|\tilde x)- \mu(.|\tilde x)\|_1^2 + \rho^{\pi^*_\mu}(\tilde x)\sum_{a\in A}[\pi^*_\mu(a|\tilde x)-\hat \pi_\alpha(a|\tilde x)]Q^i_{\mu,\pi^*_\mu}(\tilde x,a)]\textrm{ by the Pinsker inequality.}\\
    &\leq \eta \alpha^2 \left(\sum \limits_{h \in \tilde x} \rho^{\pi^{*i}_\mu}(h)\right) \frac{1}{2}\|\hat \pi(.|\tilde x)-\mu(.|\tilde x)\|_1^2 + \alpha \rho^{\pi^*_\mu}(\tilde x)\underbrace{\sum_{a\in A}[\pi^*_\mu(a|\tilde x)-\hat \pi(a|\tilde x)]Q^i_{\mu,\pi^*_\mu}(\tilde x,a)]}_{<0 \textrm{ as $Q^i_{\mu,\pi^*_\mu}(\tilde x,a)=Q^i_{\pi^*_\mu}(\tilde x,a)$}}
\end{align*}
So there exists $c>0$ and $d<0$ such that $V^i_{\mu,\pi^*_\mu}(h_{\textrm{init}})-V^i_{\mu,\hat \pi_\alpha}(h_{\textrm{init}})\leq c \alpha^2 + d \alpha = c \alpha (\alpha + \frac{d}{c})$. And finally there exists $\alpha>0$ such that $V^i_{\mu,\pi^*_\mu}(h_{\textrm{init}})-V^i_{\mu,\hat \pi_\alpha}(h_{\textrm{init}})<0$ which contradicts the fact that $\pi^*_\mu$ is a Nash for the game on reward $r^i_{\mu,\pi}(h, a)$.

From theorem~\ref{theorem_nash_div} we can conclude that for all $\mu \not\in \Pi^*$, we have $\Xi(F(\mu), \mu) > 0$ this directly imply that for all $\mu \not\in \Pi^*$ we have $\min_{\pi^*\in \Pi^*} \Xi(\pi^*, F(\mu)) - \min_{\pi^*\in \Pi^*} \Xi(\pi^*, \mu) < 0$.

\end{proof}

\subsection{Convergence to the Nash}
We want to show that the sequence of policies $\pi_k = F^k(\pi_0)$ converges to a Nash equilibrium of the game. And we suppose that all policies $\pi^*$ are interior.

Under these conditions, we have the following properties:
\begin{itemize}
    \item $F(.)$ is a continuous map on the interior of the simplex (see section~\ref{continuity_of_F}),
    \item for all $\mu \not\in \Pi^*$ we have $\min_{\pi^*\in \Pi^*} \Xi(\pi^*, F(\mu)) - \min_{\pi^*\in \Pi^*} \Xi(\pi^*, \mu) < 0$.
    \item $\mu \rightarrow \min_{\pi^*\in \Pi^*} \Xi(\pi^*, F(\mu))$ is a positive function infinite on the border of the simplex continuous in $\mu$.
    \item $\mu \rightarrow \Delta V (\mu) = \min_{\pi^*\in \Pi^*} \Xi(\pi^*, F(\mu)) - \min_{\pi^*\in \Pi^*} \Xi(\pi^*, \mu)$ is continuous in $\mu$.
\end{itemize}

Let us write $\bar\Omega_c = \{\mu|\; \min_{\pi^*\in \Pi^*} \Xi(\pi^*, \mu) \leq c\}$ and $\Omega_c = \{\mu|\; \min_{\pi^*\in \Pi^*} \Xi(\pi^*, \mu) \leq c\}$. For all finite $c$ the set $\Omega_c$ is closed and bounded set (and $\bar\Omega_c$ is an open bounded set). Then $\Omega_c$ is a compact set.

Let us consider that $C = \min_{\pi^*\in \Pi^*} \Xi(\pi^*, \pi_0)$. Since $\Delta V (\mu)<0$ then the sequence of $\min_{\pi^*\in \Pi^*} \Xi(\pi^*, \pi_k)$ converges to $c$.

By contradiction let us suppose that $c>0$. This means that all the $\pi_k$ are all in the closed set $\Omega_{C,c} = \bar\Omega_C \backslash \Omega_c$. The set $\Omega_{C,c}$ is bounded and thus is a compact. The image of $\Omega_{C,c}$ through $\Delta V (.)$ (which is a continuous map) is a compact $K$ and $k_{\max} = \sup_{x \in K}<0$. This means that $\forall k, \Delta V (\pi_k)\leq k_{\max}$

This means that
\begin{align*}
    \min_{\pi^*\in \Pi^*} \Xi(\pi^*, \pi_k) &\leq \min_{\pi^*\in \Pi^*} \Xi(\pi^*, \pi_0) + \sum \limits_{i=0}^{k-1}\Delta V (\pi_i)\\
    &\leq C + k\times k_{\max}
\end{align*}
This contradicts the fact that $c>0$ as there exists $k$ such that $C + k\times k_{\max}<c$.

In the end $k \rightarrow \min_{\pi^*\in \Pi^*} \Xi(\pi^*, \pi_k)$ converges to $0$ and $\pi_k$ converges to $\Pi^*$
\newpage

\section{Monotone Games}
\label{monotone_games}
In this section, we prove that zero-sum, zero-sum $N$-player polymatrix games and games where the profit of one player is decoupled from the interaction with the opponent are monotone.

(i) zero-sum two-player games, i.e., $V^1_{\pi} = -V^2_{\pi}$, the monotonicity condition becomes:

$\Omega^i(\pi, \mu) = V^i_{\pi^i, \pi^{-i}} (h_{\textrm{init}}) - V^i_{\mu^i, \pi^{-i}} (h_{\textrm{init}}) - V^i_{\pi^i, \mu^{-i}} (h_{\textrm{init}}) + V^i_{\mu^i, \mu^{-i}} (h_{\textrm{init}})$

It is easy to notice that $\Omega^1(\pi, \mu) = -\Omega^2(\pi, \mu) $ as $\forall \pi, \mu$ $V^1_{\mu^1, \pi^{2}} (h_{\textrm{init}}) = -V^2_{\mu^1, \pi^{2}} (h_{\textrm{init}})$

(ii) in zero-sum $N$-player polymatrix games, i.e., when the value can be decomposed in a sum of pairwise interactions $V^i_{\pi} = \sum \limits_{j \neq i} \tilde V^i_{\pi^i, \pi^j}$ with $\tilde V^i_{\pi^i, \pi^j} = -\tilde V^j_{\pi^j, \pi^i}$

In that case:
\begin{align*}
    \sum \limits_{i\in \{1,\dots,N\}} \Omega^i(\pi, \mu) &= \sum \limits_{i\in \{1,\dots,N\}} \sum \limits_{j \neq i} [\tilde V^i_{\pi^i, \pi^j}(h_{\textrm{init}}) - V^i_{\pi^i, \mu^j}(h_{\textrm{init}}) - V^i_{\mu^i, \pi^j}(h_{\textrm{init}}) + V^i_{\mu^i, \mu^j}(h_{\textrm{init}})]\\
    &\sum \limits_{i\in \{1,\dots,N\}} \sum \limits_{j < i} [\tilde V^i_{\pi^i, \pi^j}(h_{\textrm{init}}) - V^i_{\pi^i, \mu^j}(h_{\textrm{init}}) - V^i_{\mu^i, \pi^j}(h_{\textrm{init}}) + V^i_{\mu^i, \mu^j}(h_{\textrm{init}})]\\
    &\qquad\qquad\qquad\qquad + [\tilde V^j_{\pi^j, \pi^i}(h_{\textrm{init}}) - V^j_{\pi^j, \mu^i}(h_{\textrm{init}}) - V^j_{\mu^j, \pi^i}(h_{\textrm{init}}) + V^j_{\mu^j, \mu^i}(h_{\textrm{init}})]\\
    =\sum \limits_{i\in \{1,\dots,N\}} \sum \limits_{j < i}  0 = 0
\end{align*}

(iii) in games where the profit of one player is decoupled from the interaction with the opponents, i.e., when the value can be decomposed in $V^i_{\pi} = \bar V^i_{\pi^i} + \bar V^i_{\pi^{-i}}$.

In this case $\Omega^i(\pi, \mu) = V^i_{\pi^i, \pi^{-i}} (h_{\textrm{init}}) - V^i_{\mu^i, \pi^{-i}} (h_{\textrm{init}}) - V^i_{\pi^i, \mu^{-i}} (h_{\textrm{init}}) + V^i_{\mu^i, \mu^{-i}} (h_{\textrm{init}}) = \bar V^i_{\pi^i} + \bar V^i_{\pi^{-i}} - [\bar V^i_{\mu^i} + \bar V^i_{\pi^{-i}}] - [\bar V^i_{\pi^i} + \bar V^i_{\mu^{-i}}] + \bar V^i_{\mu^i} + \bar V^i_{\mu^{-i}} = 0$

\newpage
\section{Empirical set up}
\label{empirical_setup}

\begin{figure}[!htb]
\rule{\linewidth}{1.0pt}
\begin{tabular}{ l l }
  batch size & the batch size used to update the policy and the $Q$-function \\
  batch size actor & to fasten the experiment we use an actor critic setup with 512 actors all\\
  & of them produces batch of trajectories. \\
  eta & the parameter used to transform the reward \\
  lambda Retrace & the parameter of retrace \\
  epsilon greedy & the epsilon greedy policy parameter \\
  Gradient clipping value & all gradient on the policy and the value are clipped \\
  max number of steps & the maximum number of steps \\
  reward re-centered every  & the recentering periode \\
  policy learning rate start & the policy learning rate starts its exponential decay at this value \\
  policy learning rate end & the exponential decay ends at this value \\
  threshold NeuRD & the NeuRD threshold \\
  Neural net structure for $\pi$ and $Q$  & MLP with 2 hidden layer of 128 unit\\
  \bottomrule
\end{tabular}
\caption{This table summarize the meaning of all hyperparameters of the algorithm}
\end{figure}
\subsection{Estimation of the Critique}
The update on a $Q$-function is done such as to minimize the $l_2$-norm between $\hat Q^i_{\bm{w}}(x_l, a_l)$ a retrace target~\cite{espeholt2018impala} constructed using the sequence or policies and rewards (written $Q^i_{\textrm{retrace target}}$). 
\begin{align}
w_i \leftarrow w_i + \alpha \sum_{l=0}^K \tau(x_l) \times \Big[ Q^i_{\textrm{retrace target}}(x_l, a_l) - \hat Q^i_{\bm{w}}(x_l, a_l) \Big] \partial_{w_i}\hat Q^i_{\bm{w}}(x_l, a_l).
\end{align}
\subsection{Low Variance Unbiased Estimate of the Expected Payoff}
This version of NeuRD uses an unbiased estimate and low-variance of the return~\citep{schmid2018variance}. We will account that the policy $\pi^i_{\theta_i}(a^i|x_l)$ we want to evaluate can be different from the one we are sampling $\nu^i_{\theta_i}(a^i|x_l)$ and the unbiased return is computed as follow in the case of the reward transform for zero-sum games as follows:
\begin{align*}
    &\bar Q^i_{\bm{w}}(x_l, a)\\
    &= \left\{
    \begin{array}{ll}
        - \eta \log(\pi(a|x_l)) + \hat Q^i_{\bm{w}}(x_l, a) & \mbox{if } a \neq a_l \\
        &\\
        - \eta \log(\pi(a|x_l)) + \hat Q^i_{\bm{w}}(x_l, a)\\
        + \frac{1}{\nu(a^i|x_l)} \left[r^i(x_l, a_l) + E_{b \sim \pi(.|x_{l+1})} [\bar Q^i_{\bm{w}}(x_{l+1}, b) ] - \hat Q^i_{\bm{w}}(x_l, a) \right]& \mbox{if } a = a_l
    \end{array}
\right\}
\end{align*}
if $\tau(x_l)=i$
\begin{align*}
    &\bar Q^i_{\bm{w}}(x_l, a) = \left\{
    \begin{array}{ll}
        \eta \log(\pi(a|x_l)) & \mbox{if } a \neq a_l \\
        \eta \log(\pi(a|x_l)) + \frac{1}{\nu(a^i|x_l)} \left[r^i(x_l, a_l) + E_{b \sim \pi(.|x_{l+1})} [\bar Q^i_{\bm{w}}(x_{l+1}, b) ]\right]& \mbox{if } a = a_l
    \end{array}
\right\}
\end{align*}
if $\tau(x_l)\neq i$

By convention, we will have that $\forall a, \bar Q^i_{\bm{w}}(x_{K+1},a) = 0$

\subsection{NeuRD update}
In the second step we correct for the reach probability:
$$\tilde Q^i_{\bm{w}}(x_l, a) = \left(\prod \limits_{k=0}^{l-1} \frac{\mathbf{1}_{\tau{x_l} \neq i}\pi(a_k|x_l) + \mathbf{1}_{\tau{x_l} = i}}{\nu(a_k|x_l)} \right) \bar Q^i_{\bm{w}}(x_l, a)$$
Where $\pi$ and $\nu$ are the policies at the player's turn.

The policy update follows the following equation:
\begin{align}
\theta_i \leftarrow \theta_i + \alpha \sum_{l=0}^K \mathbf{1}_{\tau{x_l} = i} \sum_{a^i} \partial_{\theta_i} \xi^i_{\theta_i}(a^i|x_l)  \left[\tilde Q^i_{\bm{w}}(x_l, a^i)\right].
\end{align}

Where $\xi^i_{\theta_i}$ is the logit of policy and $\textrm{softmax}(\xi^i_{\theta_i}) = \pi^i_{\theta_i}$. The NeuRD update rule require an additional clipping parameter to avoid numerical instabilities. We leave the reader to~\cite{omidshafiei2019neural}.

Last, we obtained the exploration policy $\nu$ by doing an epsilon greedy policy $\pi$.

\newpage
\section{Experiments}
\label{app_experiments}
\subsection{Tabular Experiments}
In this section we present experiments on Kuhn poker and Leduc poker that illustrate the convergence property for the dynamics on the transformed reward.
\subsubsection{Tabular Experiments With a Fixed Regularization (reward transformation for zero-sum games)}
The two following figures illustrate the method described in section~\ref{reward_transform}. The following experiment Shows the FoReL dynamics on Kuhn poker :
\begin{figure}[!htb]
  \vspace{-10pt}
  \begin{center}
    \includegraphics[width=0.7\textwidth]{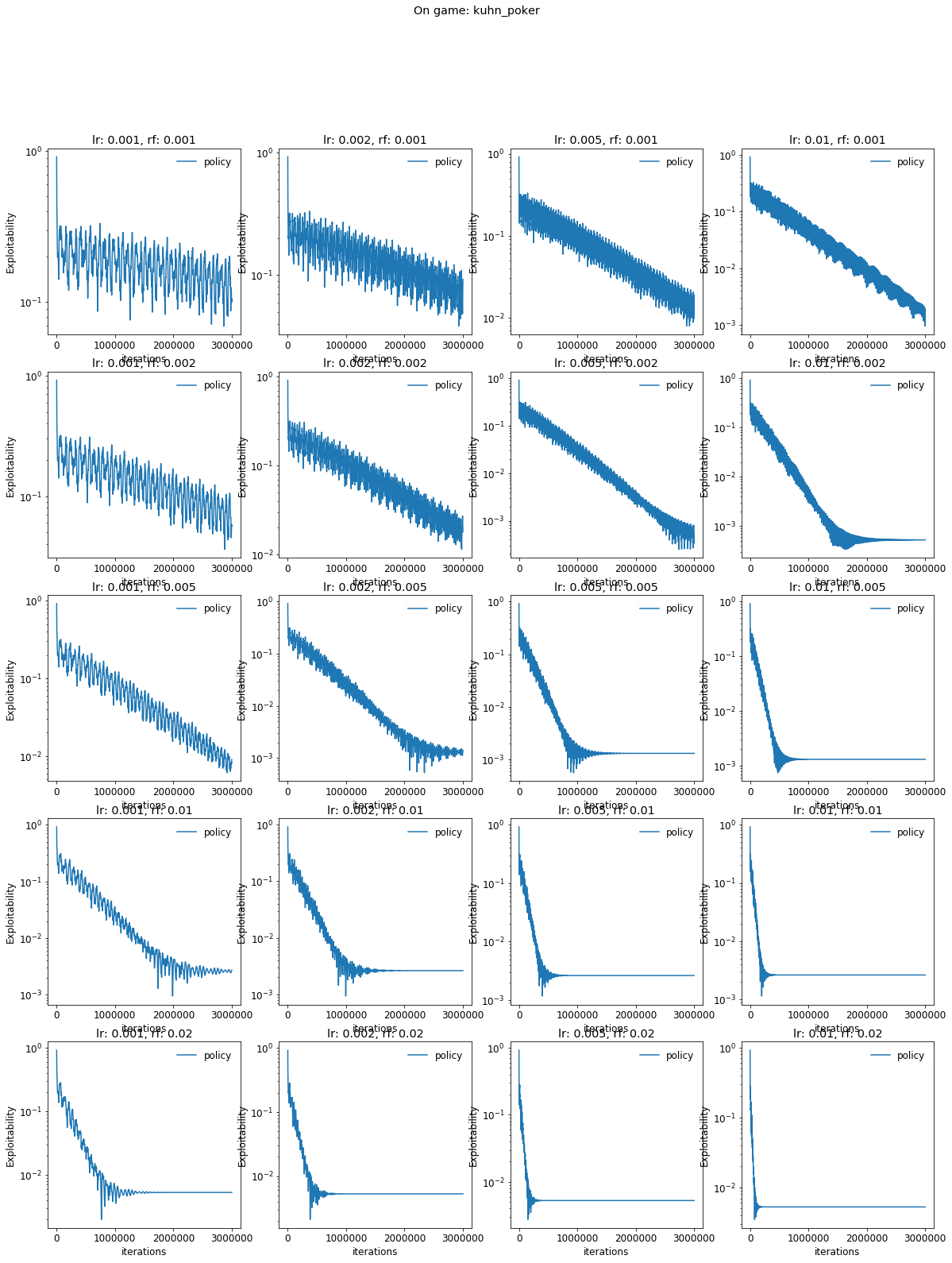}
  \end{center}
  \vspace{-10pt}
  \caption{Kuhn Poker.}
  \vspace{-0pt}
  \label{NeuRD_fixed_regularization}
\end{figure}
\newpage
And the following experiment Shows the FoReL dynamics on Leduc poker :
\begin{figure}[!htb]
  \vspace{-10pt}
  \begin{center}
    \includegraphics[width=0.8\textwidth]{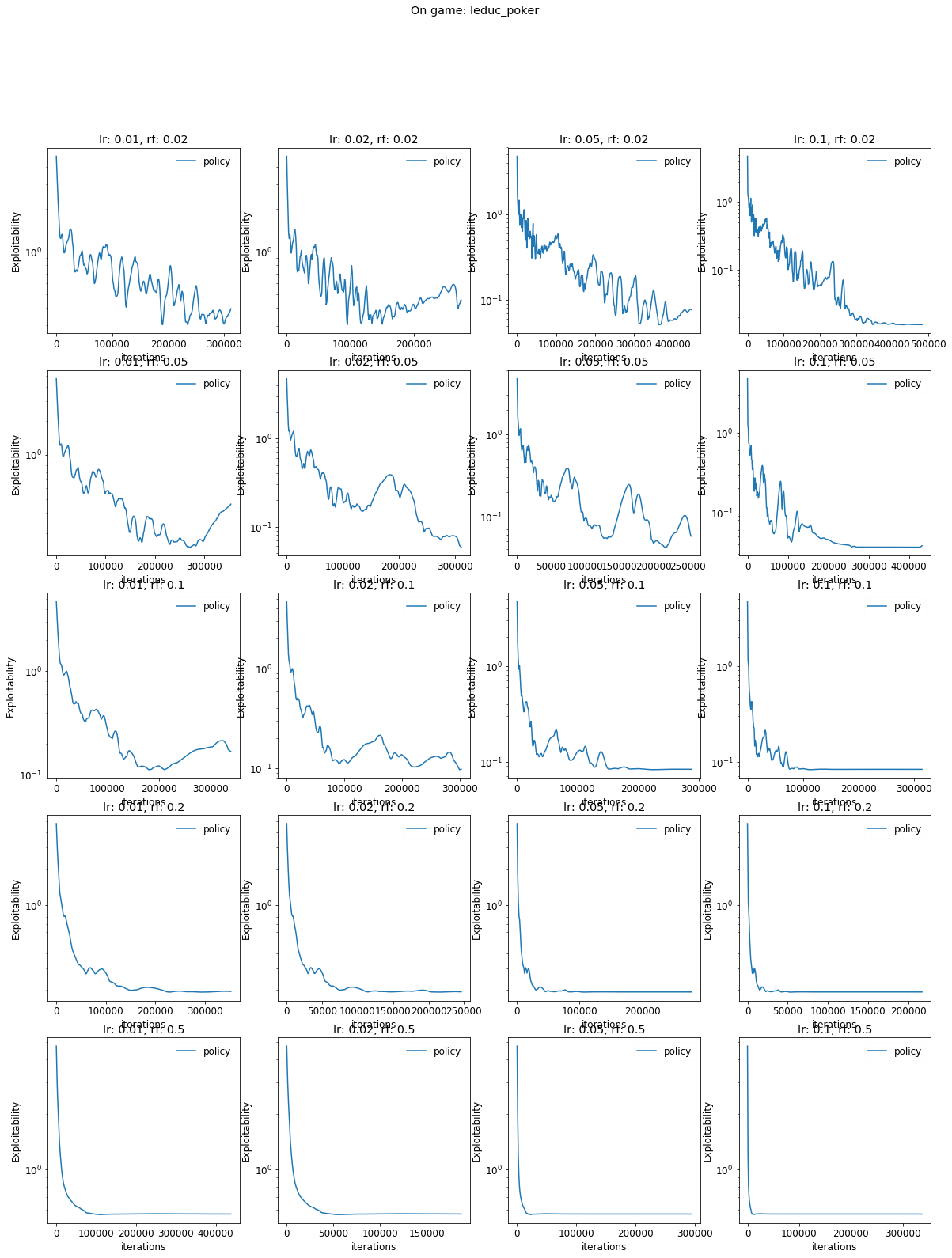}
  \end{center}
  \vspace{-10pt}
  \caption{Leduc Poker.}
  \vspace{-0pt}
  \label{NeuRD_fixed_regularization}
\end{figure}

\newpage
\subsubsection{Tabular Experiments With a Fixed Regularization (reward transformation for monotone games)}
The two following figures illustrate the method described in section~\ref{reward_transform}. The following experiment Shows the FoReL dynamics on Kuhn poker :
\begin{figure}[!htb]
  \vspace{-10pt}
  \begin{center}
    \includegraphics[width=0.9\textwidth]{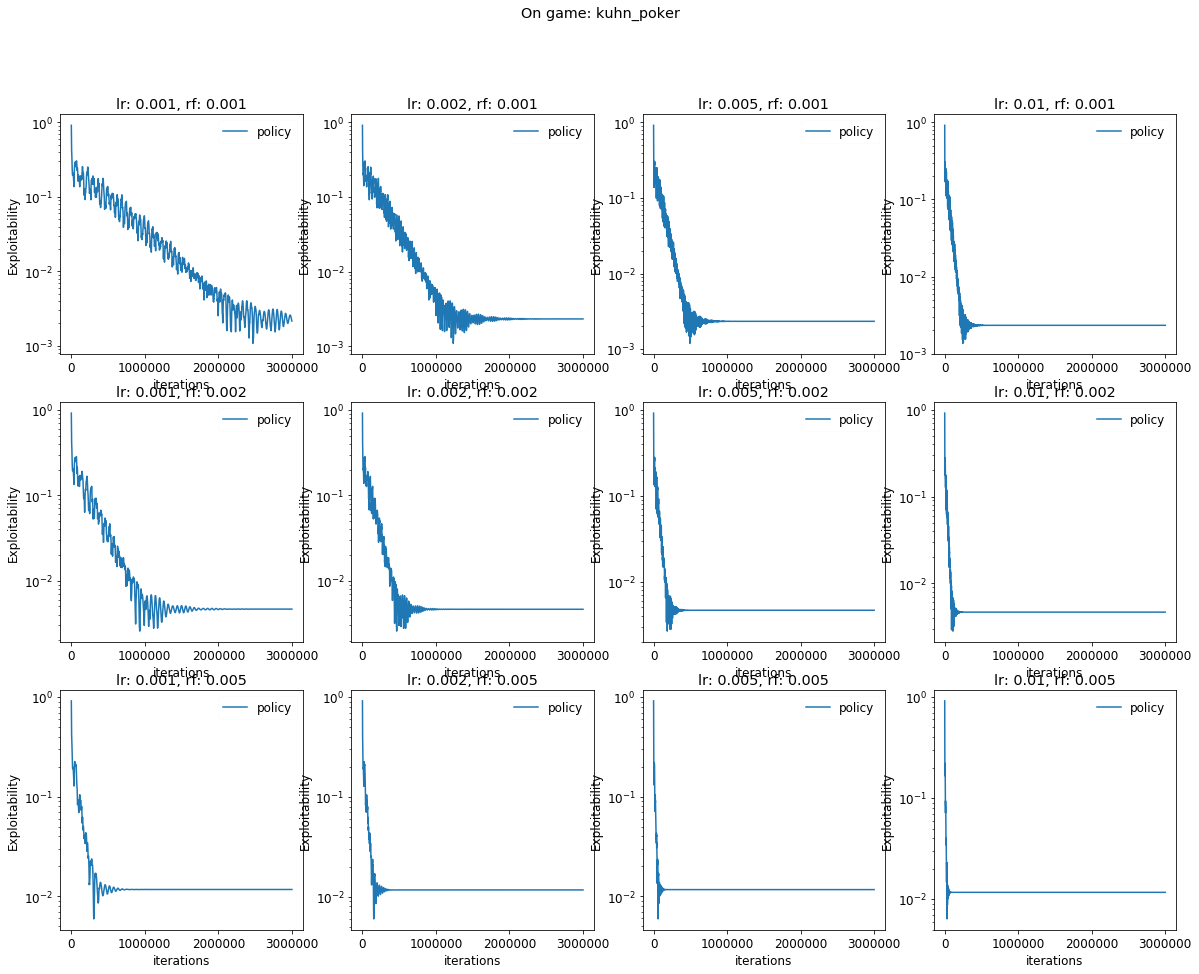}
  \end{center}
  \vspace{-10pt}
  \caption{Kuhn Poker.}
  \vspace{-0pt}
  \label{NeuRD_fixed_regularization}
\end{figure}
\newpage
And the following experiment Shows the FoReL dynamics on Leduc poker :
\begin{figure}[!htb]
  \vspace{-10pt}
  \begin{center}
    \includegraphics[width=0.9\textwidth]{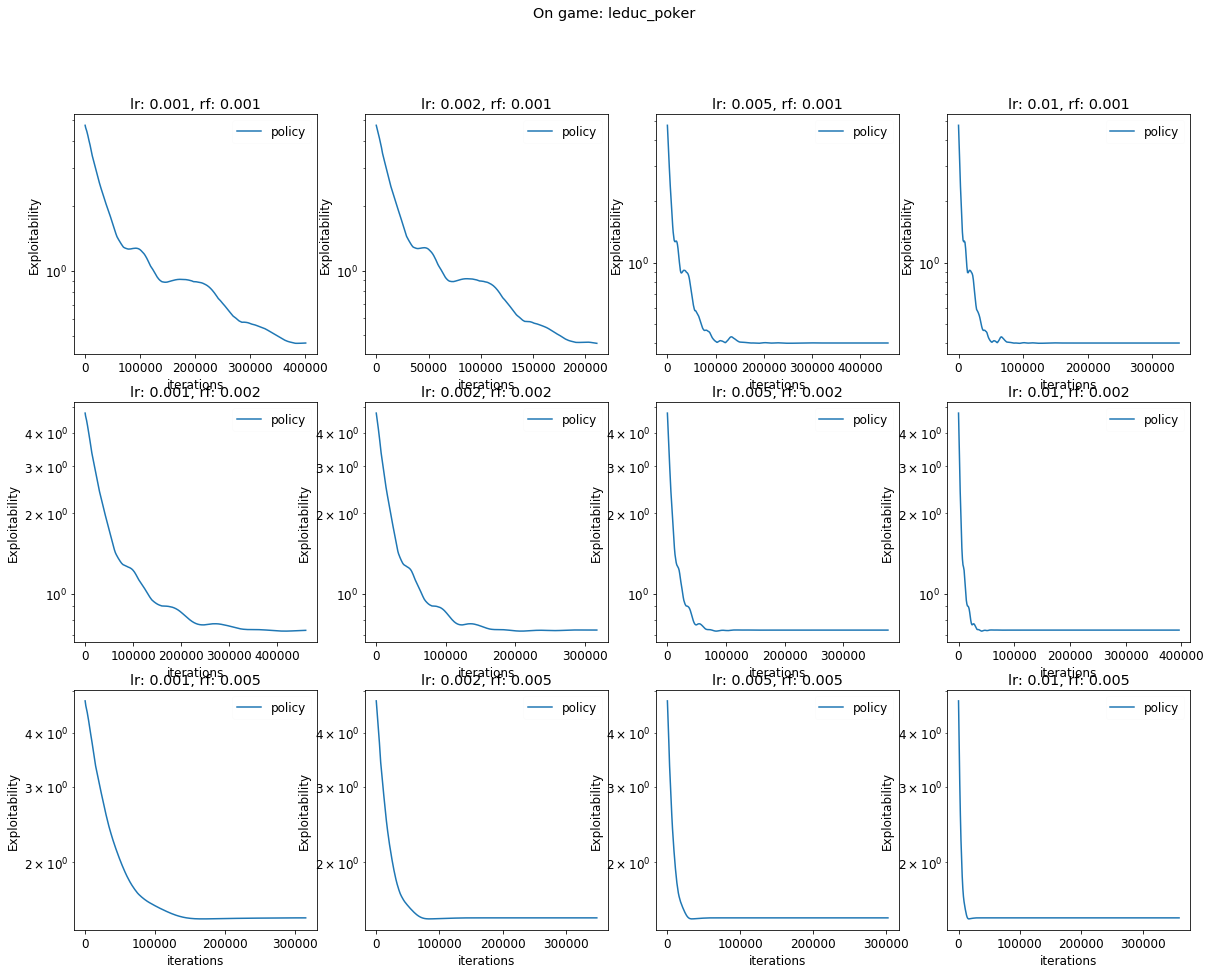}
  \end{center}
  \vspace{-10pt}
  \caption{Leduc Poker.}
  \vspace{-0pt}
  \label{NeuRD_fixed_regularization}
\end{figure}

\newpage
\subsubsection{Tabular Experiments With an addaptive Regularization (reward transformation for monotone games and the reward is changed every $20000$ steps)}
The two following figures illustrate the method described in section~\ref{convergence_to_an_equilibrium}. And the following experiment Shows the FoReL dynamics on Kuhn poker :
\begin{figure}[!htb]
  \vspace{-10pt}
  \begin{center}
    \includegraphics[width=0.8\textwidth]{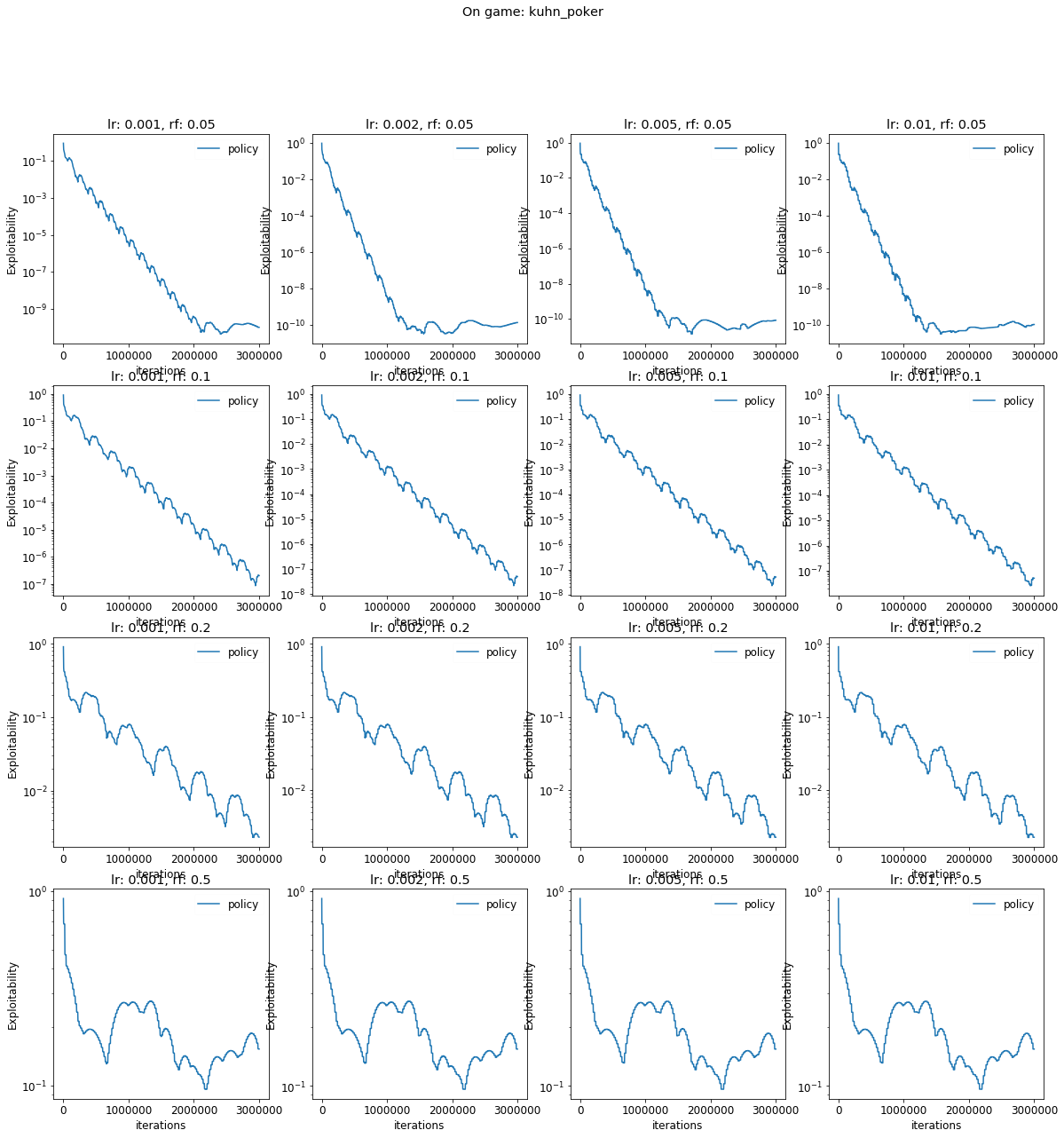}
  \end{center}
  \vspace{-10pt}
  \caption{Kuhn Poker.}
  \vspace{-0pt}
  \label{NeuRD_fixed_regularization}
\end{figure}
\newpage
And the following experiment Shows the FoReL dynamics on Leduc poker :
\begin{figure}[!htb]
  \vspace{-10pt}
  \begin{center}
    \includegraphics[width=0.8\textwidth]{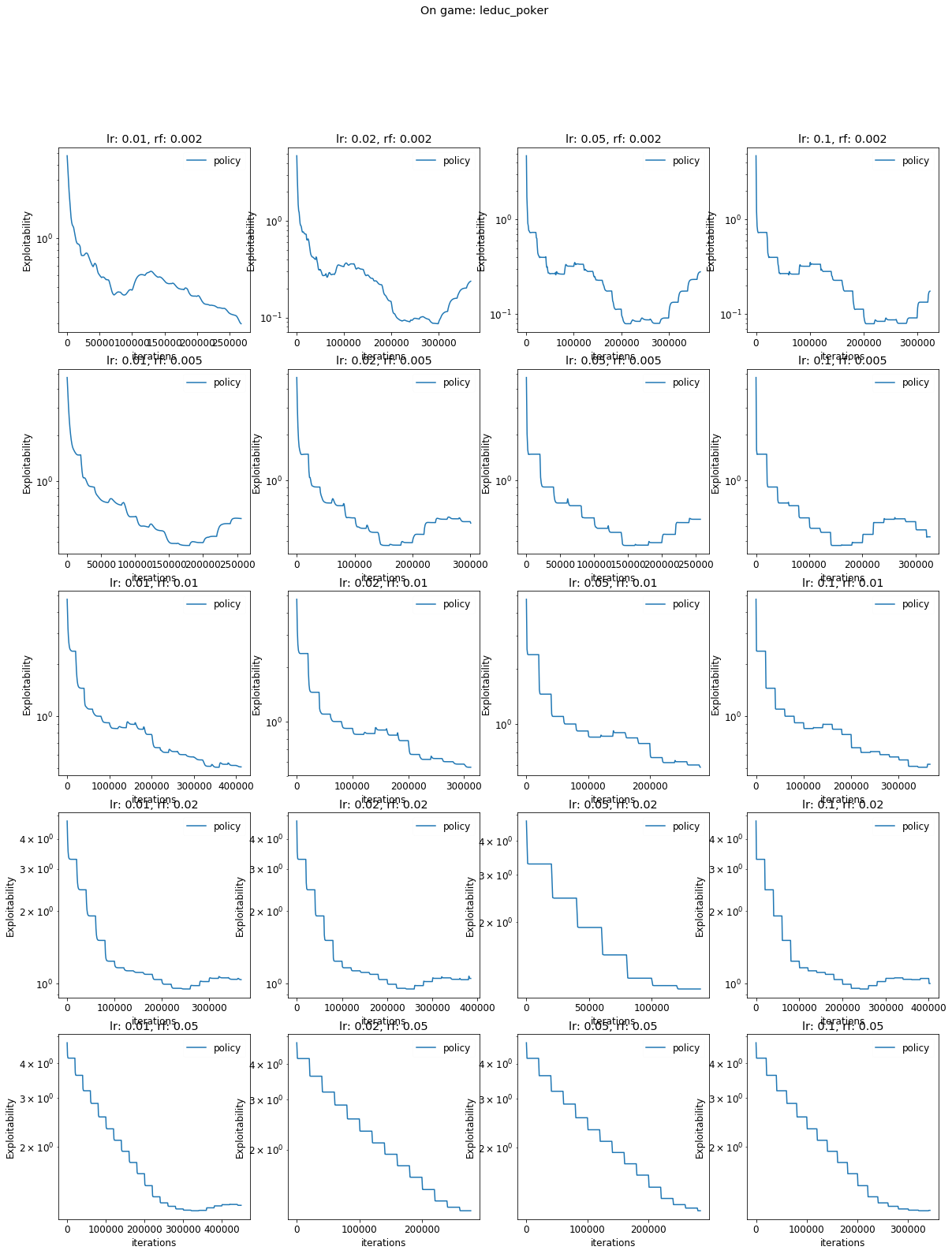}
  \end{center}
  \vspace{-10pt}
  \caption{Leduc Poker.}
  \vspace{-0pt}
  \label{NeuRD_fixed_regularization}
\end{figure}

\newpage
\subsection{Deep Reinforcement Learning Experiments with player only regularization}
In this section, we run NeuRD on Leduc poker, Kuhn poker, Liars Dice and GoofSpiel with the reward transform for monotone games. The reward is adapted every 75000 steps.

\begin{figure}[!htb]
\rule{\linewidth}{1.0pt}
\begin{tabular}{ l l }
  batch size & 256 \\
  batch size actor & 32 \\
  eta & $\{1.0, 0.5, 0.2, 0.05, 0.02, 0.0\}$ \\
  lambda Retrace & 1.0 \\
  epsilon greedy & 0.1 \\
  Gradient clipping value & 1000 \\
  max number of steps & 4000000 \\
  reward re-centered every  & 75000 \\
  policy learning rate start & 0.01 \\
  policy learning rate end & 0.00001 \\
  threshold NeuRD & 2 \\
  Neural net structure for $\pi$ and $Q$  & MLP with 2 hidden layer of 128 unit\\
  \bottomrule
\end{tabular}
\end{figure}

Experiment on Liars Dice:

\begin{figure}[!htb]
  \vspace{-10pt}
  \begin{center}
    \includegraphics[width=0.9\textwidth]{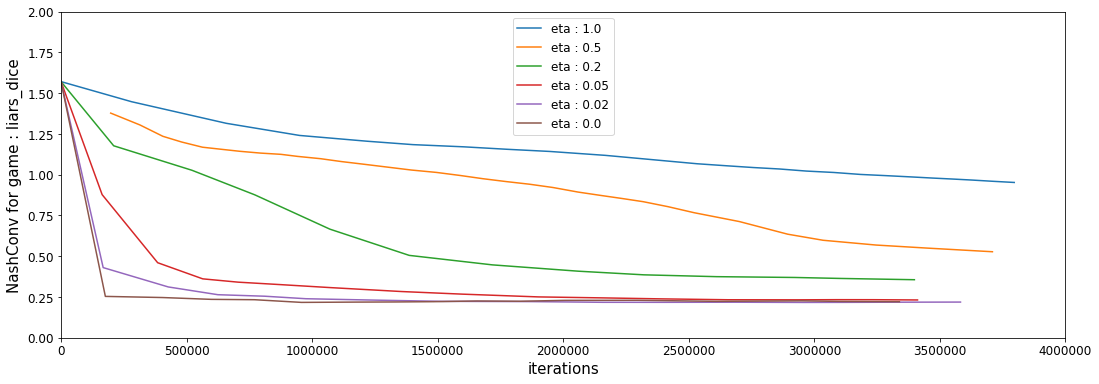}
  \end{center}
  \vspace{-10pt}
  \caption{Liars Dice.}
  \vspace{-0pt}
  \label{NeuRD_fixed_regularization}
\end{figure}

Experiment on Leduc Poker:

\begin{figure}[!htb]
  \vspace{-10pt}
  \begin{center}
    \includegraphics[width=0.9\textwidth]{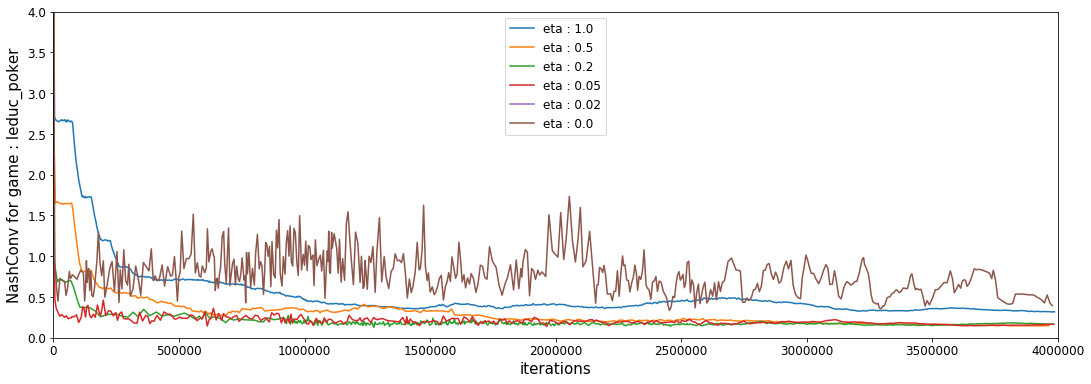}
  \end{center}
  \vspace{-10pt}
  \caption{Leduc Poker.}
  \vspace{-0pt}
  \label{NeuRD_fixed_regularization}
\end{figure}

\newpage
Experiment on Kuhn Poker:

\begin{figure}[!htb]
  \vspace{-10pt}
  \begin{center}
    \includegraphics[width=0.9\textwidth]{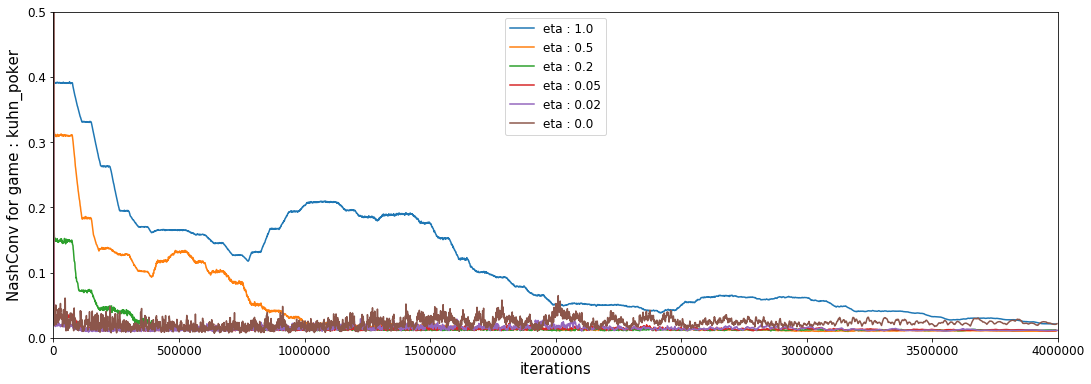}
  \end{center}
  \vspace{-10pt}
  \caption{Kuhn Poker.}
  \vspace{-0pt}
  \label{NeuRD_fixed_regularization}
\end{figure}

Experiment on Goofspiel:

\begin{figure}[!htb]
  \vspace{-10pt}
  \begin{center}
    \includegraphics[width=0.9\textwidth]{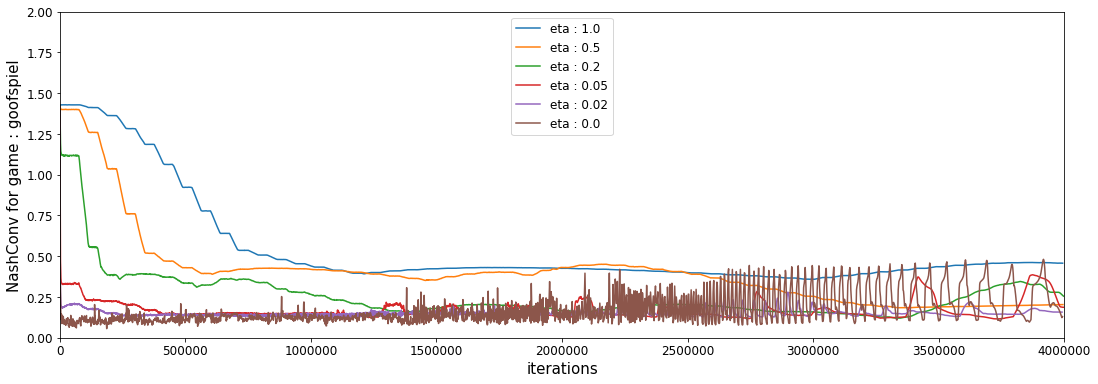}
  \end{center}
  \vspace{-10pt}
  \caption{Goofspiel $4$ cards.}
  \vspace{-0pt}
  \label{NeuRD_fixed_regularization}
\end{figure}

\newpage
In these experiments, we run NeuRD on Leduc poker, Kuhn poker, Liars Dice and GoofSpiel with the reward transform for monotone games with a constant regularization.

\begin{figure}[!htb]
\rule{\linewidth}{1.0pt}
\begin{tabular}{ l l }
  batch size & 256 \\
  batch size actor & 32 \\
  eta & $\{1.0, 0.5, 0.2, 0.05, 0.02, 0.0\}$ \\
  lambda Retrace & 1.0 \\
  epsilon greedy & 0.1 \\
  Gradient clipping value & 1000 \\
  max number of steps & 4000000 \\
  reward re-centered every  & Never \\
  policy learning rate start & 0.01 \\
  policy learning rate end & 0.00001 \\
  threshold NeuRD & 2 \\
  Neural net structure for $\pi$ and $Q$  & MLP with 2 hidden layer of 128 unit\\
  \bottomrule
\end{tabular}
\end{figure}

Experiment on Liars Dice:

\begin{figure}[!htb]
  \vspace{-10pt}
  \begin{center}
    \includegraphics[width=0.9\textwidth]{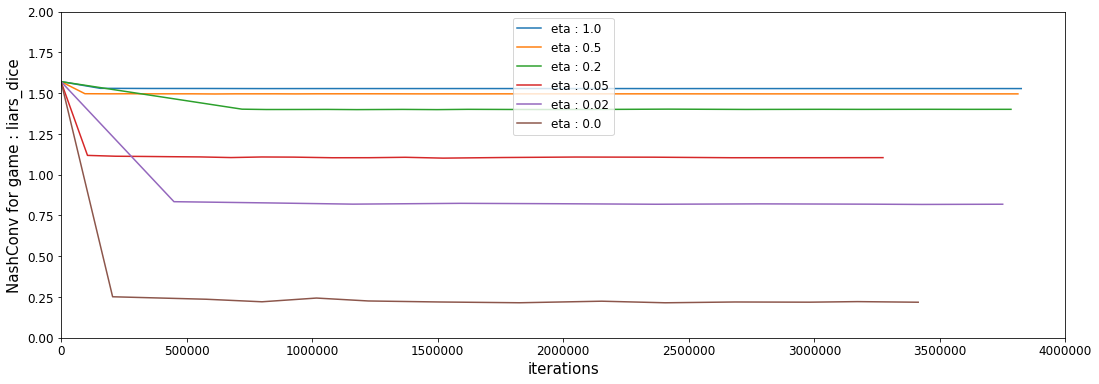}
  \end{center}
  \vspace{-10pt}
  \caption{Liars Dice.}
  \vspace{-0pt}
  \label{NeuRD_fixed_regularization}
\end{figure}

Experiment on Leduc Poker:

\begin{figure}[!htb]
  \vspace{-10pt}
  \begin{center}
    \includegraphics[width=0.9\textwidth]{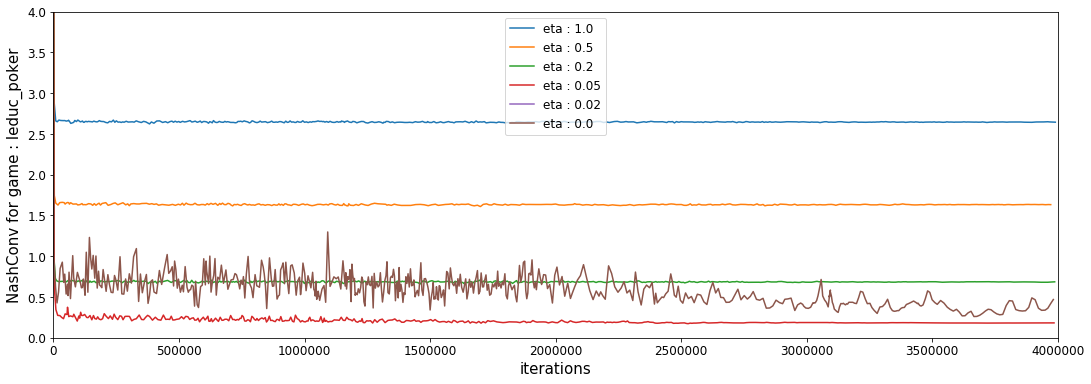}
  \end{center}
  \vspace{-10pt}
  \caption{Leduc Poker.}
  \vspace{-0pt}
  \label{NeuRD_fixed_regularization}
\end{figure}

\newpage
Experiment on Kuhn Poker:

\begin{figure}[!htb]
  \vspace{-10pt}
  \begin{center}
    \includegraphics[width=0.9\textwidth]{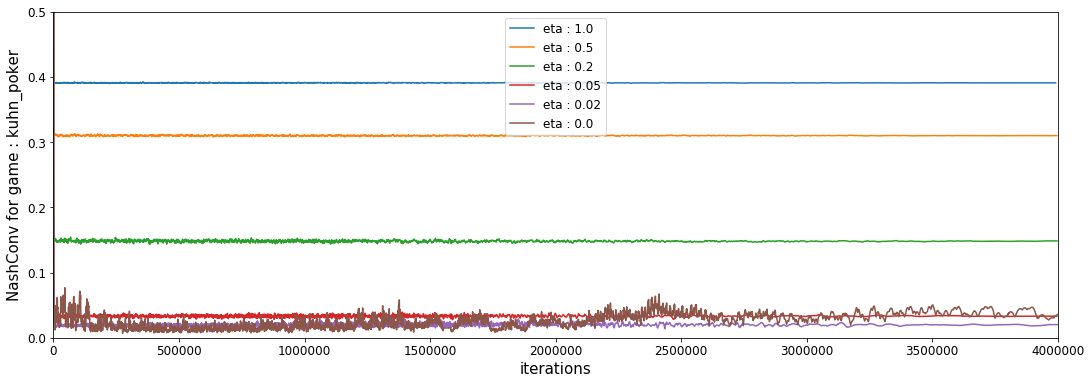}
  \end{center}
  \vspace{-10pt}
  \caption{Kuhn Poker.}
  \vspace{-0pt}
  \label{NeuRD_fixed_regularization}
\end{figure}

Experiment on Goofspiel:

\begin{figure}[!htb]
  \vspace{-10pt}
  \begin{center}
    \includegraphics[width=0.9\textwidth]{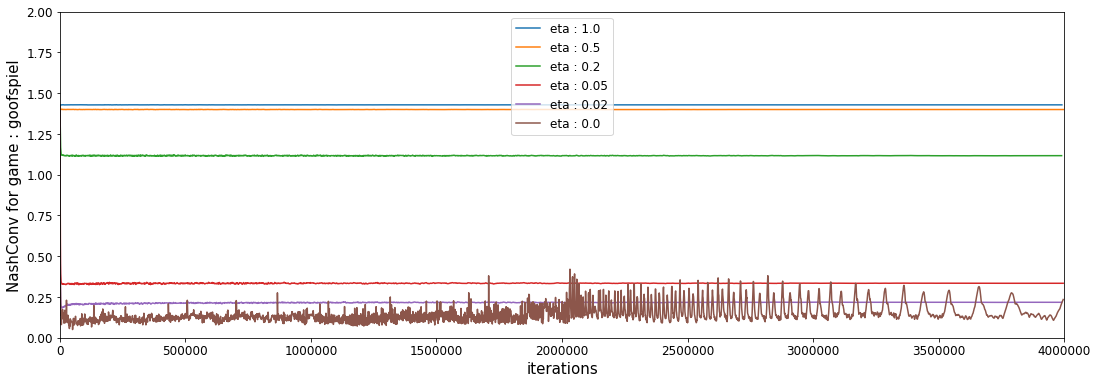}
  \end{center}
  \vspace{-10pt}
  \caption{Goofspiel $4$ cards.}
  \vspace{-0pt}
  \label{NeuRD_fixed_regularization}
\end{figure}

\newpage
In these experiments, we run NeuRD on Leduc poker, Kuhn poker, Liars Dice and GoofSpiel with the reward transform for monotone games with an exponential decay regularization to the regularization on the label.

\begin{figure}[!htb]
\rule{\linewidth}{1.0pt}
\begin{tabular}{ l l }
  batch size & 256 \\
  batch size actor & 32 \\
  eta exponential decay starting from $1.0$ until the target value & $\{1.0, 0.5, 0.2, 0.05, 0.02, 0.0\}$ \\
  lambda Retrace & 1.0 \\
  epsilon greedy & 0.1 \\
  Gradient clipping value & 1000 \\
  max number of steps & 4000000 \\
  reward re-centered every  & Never \\
  policy learning rate start & 0.01 \\
  policy learning rate end & 0.00001 \\
  threshold NeuRD & 2 \\
  Neural net structure for $\pi$ and $Q$  & MLP with 2 hidden layer of 128 unit\\
  \bottomrule
\end{tabular}
\end{figure}

Experiment on Liars Dice:

\begin{figure}[!htb]
  \vspace{-10pt}
  \begin{center}
    \includegraphics[width=0.9\textwidth]{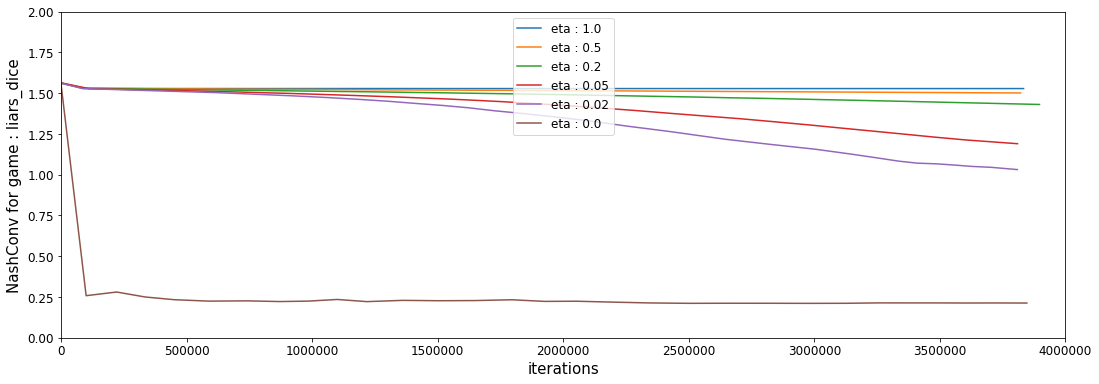}
  \end{center}
  \vspace{-10pt}
  \caption{Liars Dice.}
  \vspace{-0pt}
  \label{NeuRD_fixed_regularization}
\end{figure}

Experiment on Leduc Poker:

\begin{figure}[!htb]
  \vspace{-10pt}
  \begin{center}
    \includegraphics[width=0.9\textwidth]{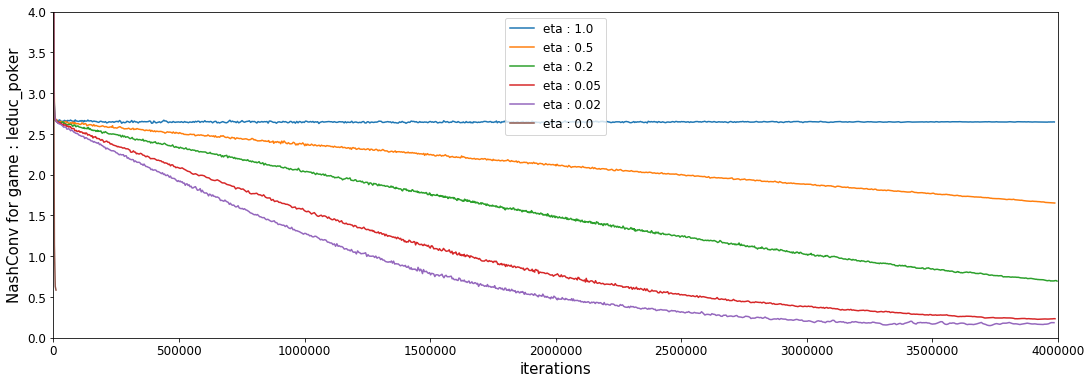}
  \end{center}
  \vspace{-10pt}
  \caption{Leduc Poker.}
  \vspace{-0pt}
  \label{NeuRD_fixed_regularization}
\end{figure}

\newpage
Experiment on Kuhn Poker:

\begin{figure}[!htb]
  \vspace{-10pt}
  \begin{center}
    \includegraphics[width=0.9\textwidth]{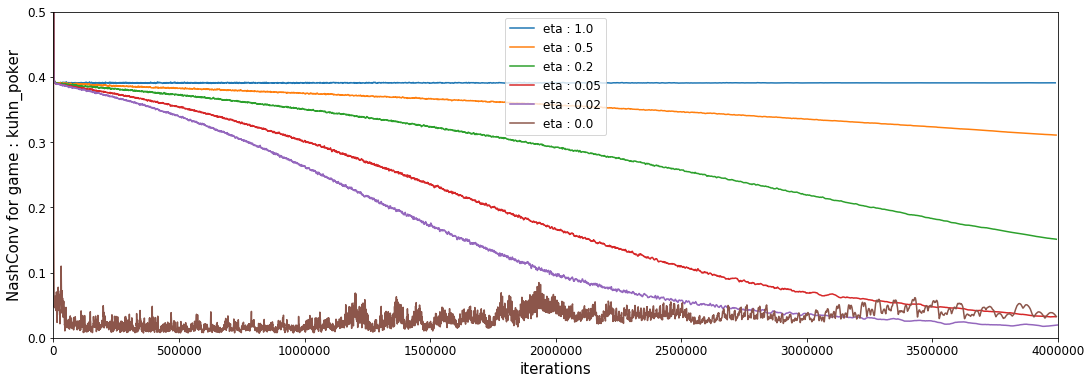}
  \end{center}
  \vspace{-10pt}
  \caption{Kuhn Poker.}
  \vspace{-0pt}
  \label{NeuRD_fixed_regularization}
\end{figure}

Experiment on Goofspiel:

\begin{figure}[!htb]
  \vspace{-10pt}
  \begin{center}
    \includegraphics[width=0.9\textwidth]{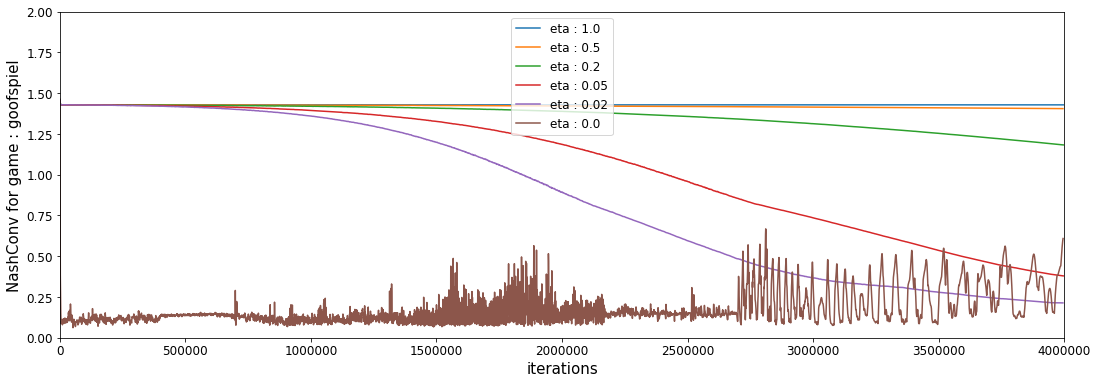}
  \end{center}
  \vspace{-10pt}
  \caption{Goofspiel $4$ cards.}
  \vspace{-0pt}
  \label{NeuRD_fixed_regularization}
\end{figure}

\newpage
\subsection{Deep Reinforcement Learning Experiments with two player regularization}
In these experiments, we run NeuRD on Leduc poker, Kuhn poker, Liars Dice and GoofSpiel with the reward transform for zero-sum games with an exponential decay regularization to the regularization on the label.

\begin{figure}[!htb]
\rule{\linewidth}{1.0pt}
\begin{tabular}{ l l }
  batch size & 256 \\
  batch size actor & 32 \\
  eta exponential decay starting from $1.0$ until the target value & $\{1.0, 0.5, 0.2, 0.05, 0.02, 0.0\}$ \\
  lambda Retrace & 1.0 \\
  epsilon greedy & 0.1 \\
  Gradient clipping value & 1000 \\
  max number of steps & 4000000 \\
  reward re-centered every  & Never \\
  policy learning rate start & 0.01 \\
  policy learning rate end & 0.00001 \\
  threshold NeuRD & 2 \\
  Neural net structure for $\pi$ and $Q$  & MLP with 2 hidden layer of 128 unit\\
  \bottomrule
\end{tabular}
\end{figure}

Experiment on Liars Dice:

\begin{figure}[!htb]
  \vspace{-10pt}
  \begin{center}
    \includegraphics[width=0.9\textwidth]{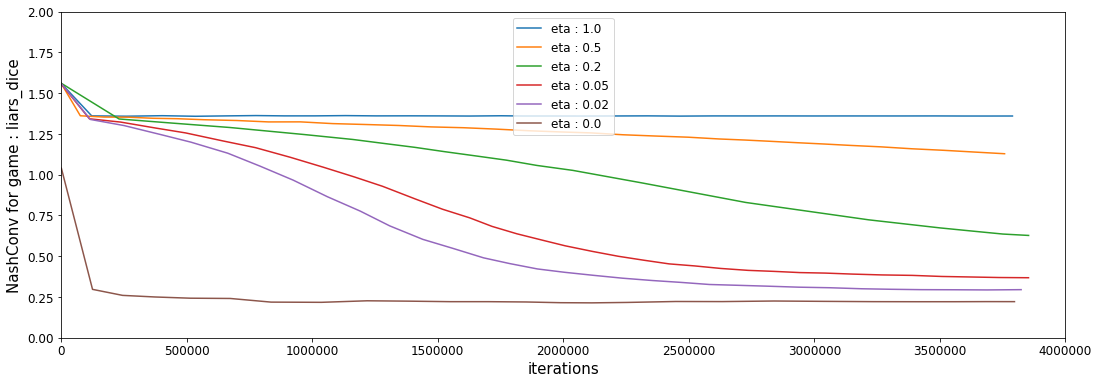}
  \end{center}
  \vspace{-10pt}
  \caption{Liars Dice.}
  \vspace{-0pt}
  \label{NeuRD_fixed_regularization}
\end{figure}

Experiment on Leduc Poker:

\begin{figure}[!htb]
  \vspace{-10pt}
  \begin{center}
    \includegraphics[width=0.9\textwidth]{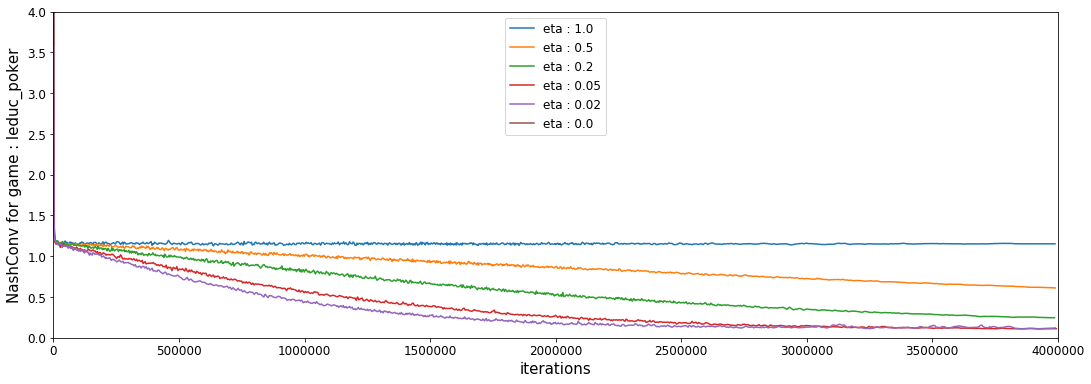}
  \end{center}
  \vspace{-10pt}
  \caption{Leduc Poker.}
  \vspace{-0pt}
  \label{NeuRD_fixed_regularization}
\end{figure}

\newpage
Experiment on Kuhn Poker:

\begin{figure}[!htb]
  \vspace{-10pt}
  \begin{center}
    \includegraphics[width=0.9\textwidth]{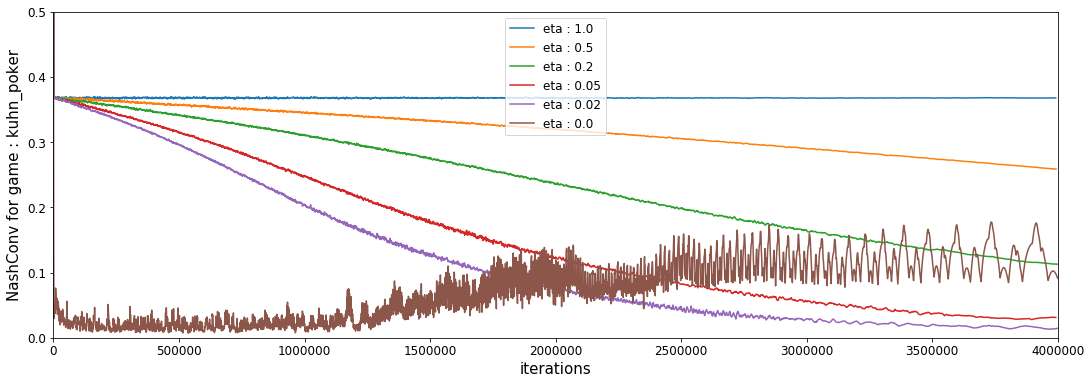}
  \end{center}
  \vspace{-10pt}
  \caption{Kuhn Poker.}
  \vspace{-0pt}
  \label{NeuRD_fixed_regularization}
\end{figure}
Experiment on Goofspiel:
\begin{figure}[!htb]
  \vspace{-10pt}
  \begin{center}
    \includegraphics[width=0.9\textwidth]{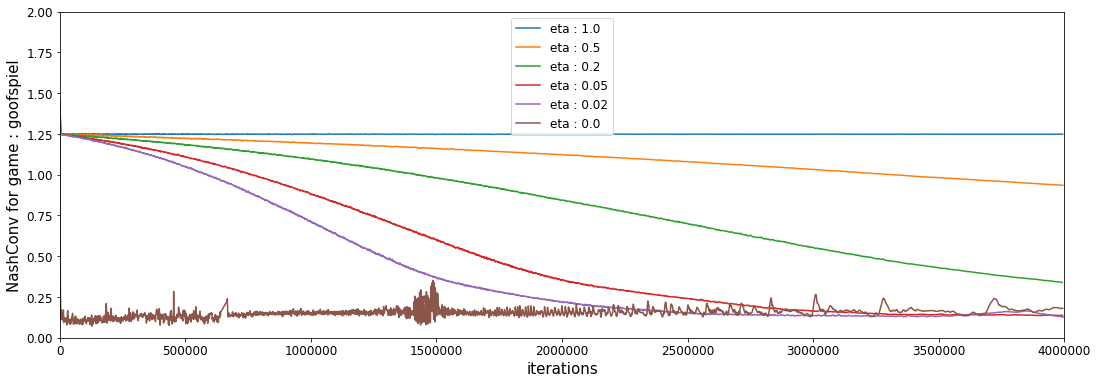}
  \end{center}
  \vspace{-10pt}
  \caption{Goofspiel $4$ cards.}
  \vspace{-0pt}
  \label{NeuRD_fixed_regularization}
\end{figure}

\newpage
\subsection{Deep Reinforcement Learning Experiments with two player regularization with a large batch}
In these experiments, we run NeuRD on Leduc poker, Kuhn poker, Liars Dice and GoofSpiel with the reward transform for zero-sum games with an exponential decay regularization to the regularization on the label.

\begin{figure}[!htb]
\rule{\linewidth}{1.0pt}
\begin{tabular}{ l l }
  batch size & 2018 \\
  batch size actor & 32 \\
  eta exponential decay starting from $1.0$ until the target value & $\{1.0, 0.5, 0.2, 0.05, 0.02, 0.0\}$ \\
  lambda Retrace & 1.0 \\
  epsilon greedy & 0.1 \\
  Gradient clipping value & 1000 \\
  max number of steps & 4000000 \\
  reward re-centered every  & Never \\
  policy learning rate start & 0.01 \\
  policy learning rate end & 0.00001 \\
  threshold NeuRD & 2 \\
  Neural net structure for $\pi$ and $Q$  & MLP with 2 hidden layer of 128 unit\\
  \bottomrule
\end{tabular}
\end{figure}

Experiment on Liars Dice:

\begin{figure}[!htb]
  \vspace{-10pt}
  \begin{center}
    \includegraphics[width=0.9\textwidth]{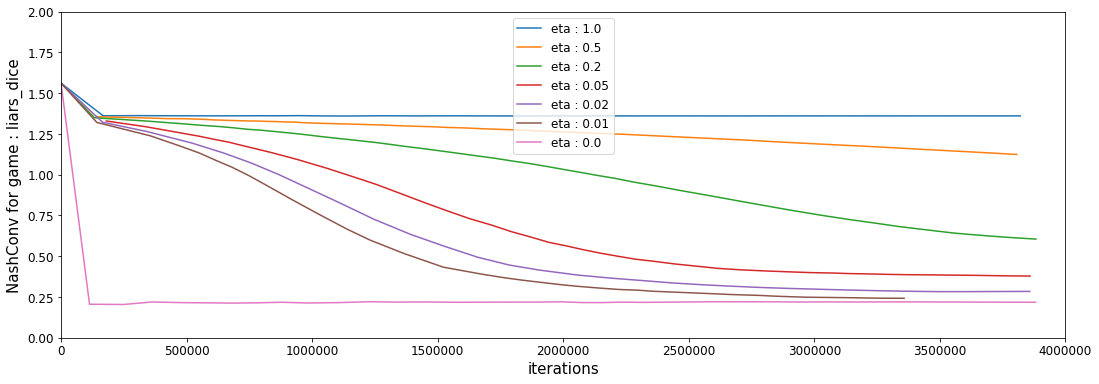}
  \end{center}
  \vspace{-10pt}
  \caption{Liars Dice.}
  \vspace{-0pt}
  \label{NeuRD_fixed_regularization}
\end{figure}

Experiment on Leduc Poker:

\begin{figure}[!htb]
  \vspace{-10pt}
  \begin{center}
    \includegraphics[width=0.9\textwidth]{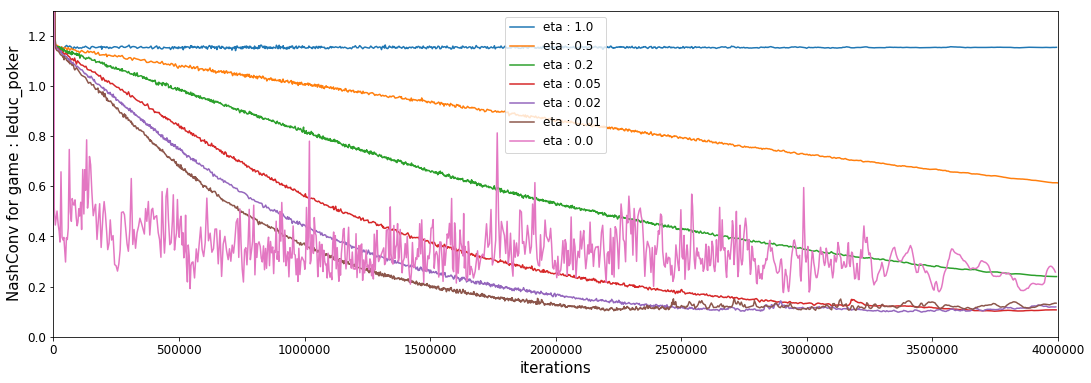}
  \end{center}
  \vspace{-10pt}
  \caption{Leduc Poker.}
  \vspace{-0pt}
  \label{NeuRD_fixed_regularization}
\end{figure}

\newpage
Experiment on Kuhn Poker:

\begin{figure}[!htb]
  \vspace{-10pt}
  \begin{center}
    \includegraphics[width=0.9\textwidth]{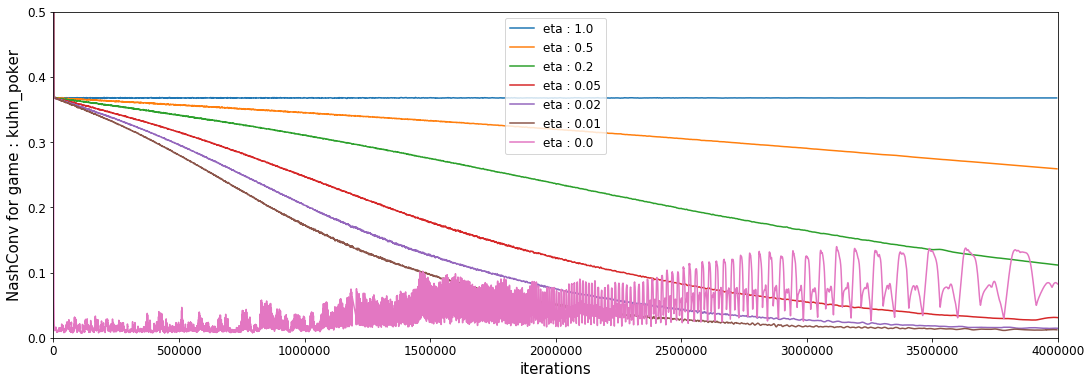}
  \end{center}
  \vspace{-10pt}
  \caption{Kuhn Poker.}
  \vspace{-0pt}
  \label{NeuRD_fixed_regularization}
\end{figure}
Experiment on Goofspiel:
\begin{figure}[!htb]
  \vspace{-10pt}
  \begin{center}
    \includegraphics[width=0.9\textwidth]{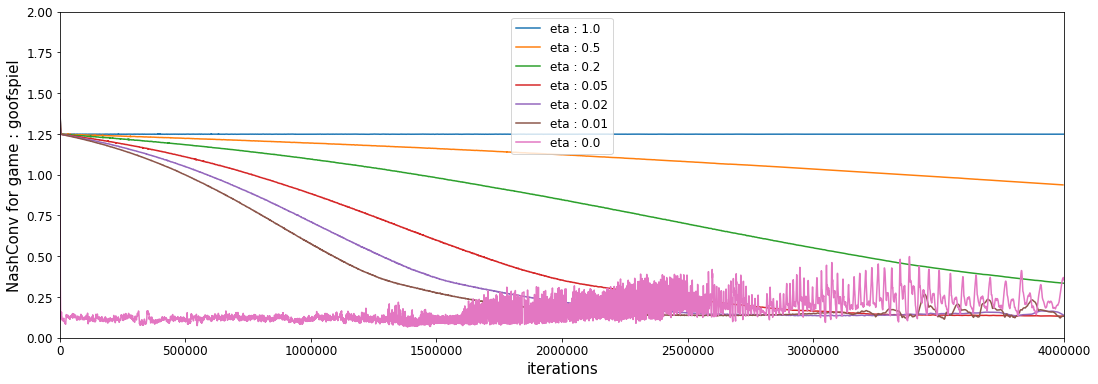}
  \end{center}
  \vspace{-10pt}
  \caption{Goofspiel $4$ cards.}
  \vspace{-0pt}
  \label{NeuRD_fixed_regularization}
\end{figure}

\end{document}